\documentclass[review,a4paper,UKenglish,cleveref, autoref, thm-restate]{lmcs}
\pdfoutput=1
\usepackage{lastpage}
\lmcsdoi{22}{2}{6}
\lmcsheading{}{\pageref{LastPage}}{}{}%
{Nov.~18,~2024}{Apr.~17,~2026}{}

\usepackage{amsthm}
\usepackage{stmaryrd}
\usepackage[utf8]{inputenc}
\usepackage{booktabs}

\theoremstyle{definition}

\title[Optimally Rewriting Formulas and Database Queries]{Optimally Rewriting Formulas and Database Queries: A Confluence of Term Rewriting, Structural Decomposition, and Complexity}

\author[H.~Chen]{Hubie Chen}[a]
\author[S.~Mengel]{Stefan Mengel}[b]

\address{Department of Informatics, King’s College London}
\address{Univ.\ Artois, CNRS, Centre de Recherche en Informatique de Lens (CRIL)}


\begin{document}

\begin{abstract}
  A central computational task in database theory, finite model theory, and computer science at large is the evaluation of a first-order sentence on a finite structure.  In the context of this task, the \emph{width} of a sentence, defined as the maximum number of free variables over all subformulas, has been established as a crucial measure, where minimizing width of a sentence (while retaining logical equivalence) is considered highly desirable.
  An undecidability result rules out the possibility of an algorithm that, given a first-order sentence, returns a logically equivalent sentence of minimum width; this result motivates the study of width minimization via syntactic rewriting rules, which is this article's focus.  For a number of common rewriting rules (which are known to preserve logical equivalence), including rules that allow for the movement of quantifiers, we present an algorithm that, given a positive first-order sentence $\phi$, outputs the minimum-width sentence obtainable from $\phi$ via application of these rules.
  We thus obtain a complete algorithmic understanding of width minimization up to the studied rules; this result is the first one---of which we are aware---that establishes this type of understanding in such a general setting.  Our result builds on the theory of term rewriting and establishes an interface among this theory, query evaluation, and structural decomposition theory.
\end{abstract}
\maketitle

\theoremstyle{plain}
\newtheorem{theorem}[thm]{Theorem}%

\newtheorem{corollary}[thm]{Corollary}

\newtheorem{proposition}[thm]{Proposition}

\newtheorem{propositionrep}[thm]{Proposition}
\newtheorem{lemma}[thm]{Lemma}

\newtheorem{lemmarep}[thm]{Lemma}

\newtheorem{claim}[thm]{Claim}

\theoremstyle{definition}
\newtheorem{example}[thm]{Example}

\newcommand{\sh}{\sharp}
\newcommand{\countp}{\mathsf{count}}

\newcommand{\mc}{\mathsf{MC}}

\newcommand{\nats}{\mathbb{N}}
\newcommand{\N}{\mathbb{N}}
\newcommand{\Q}{\mathbb{Q}}
\newcommand{\Z}{\mathbb{Z}}

\newcommand{\str}{\mathrm{str}}

\newcommand{\rela}{\mathbf{A}}
\newcommand{\relb}{\mathbf{B}}
\newcommand{\relc}{\mathbf{C}}
\newcommand{\reld}{\mathbf{D}}
\newcommand{\rele}{\mathbf{E}}
\newcommand{\relt}{\mathbf{T}}
\newcommand{\relp}{\mathbf{P}}
\newcommand{\relf}{\mathbf{F}}
\newcommand{\relg}{\mathbf{G}}
\newcommand{\relh}{\mathbf{H}}
\newcommand{\reli}{\mathbf{I}}
\newcommand{\relm}{\mathbf{M}}
\newcommand{\reln}{\mathbf{N}}
\newcommand{\relr}{\mathbf{R}}

\newcommand{\calP}{\mathcal{P}}
\newcommand{\calE}{\mathcal{E}}

\newcommand{\free}{\mathrm{free}}
\newcommand{\closed}{\mathrm{closed}}
\newcommand{\lib}{\mathrm{lib}}

\newcommand{\id}{\mathrm{id}}
\newcommand{\surj}{\mathrm{surj}}
\newcommand{\tw}{\mathsf{tw}}
\newcommand{\starsize}{\mathsf{starsize}}
\newcommand{\qaw}{\mathsf{qaw}}
\newcommand{\topp}{\mathsf{top}}
\newcommand{\atom}{\mathsf{atom}}

\newcommand{\res}{\upharpoonright}
\newcommand{\aug}{\mathsf{aug}}
\newcommand{\width}{\mathsf{width}}
\newcommand{\shwidth}{\sh\textup{-}\width}
\newcommand{\contract}{\mathsf{contract}}

\newcommand{\p}{\mathsf{P}}
\newcommand{\np}{\mathsf{NP}}

\newcommand{\FO}{\mathsf{FO}}
\newcommand{\PP}{\mathsf{PP}}
\newcommand{\EP}{\mathsf{EP}}

\newcommand{\fo}{\mathsf{FO}}
\newcommand{\pp}{\mathsf{PP}}
\newcommand{\ep}{\mathsf{EP}}

\newcommand{\pcountp}{\smallp\countp}

\newcommand{\dom}{\mathrm{dom}}
\newcommand{\powfin}{\wp_{\mathsf{fin}}}

\newcommand{\pn}[1]{\textsc{#1}}

\newcommand{\smallp}{\ensuremath{p\textup{-}}}
\newcommand{\clique}{\pn{Clique}}
\newcommand{\sclique}{\pn{\#Clique}}
\newcommand{\pclique}{\smallp\clique}
\newcommand{\psclique}{\smallp\sclique}

\newcommand{\param}[1]{\mathsf{param}\textup{-}#1}

\newcommand{\hnote}[1]{}
\newcommand{\snote}[1]{[Stefan note: #1]}

\newcommand{\lics}[1]{#1}
\newcommand{\pods}[1]{}
\newcommand{\licspods}[2]{\lics{#1}\pods{#2}}

\newcommand{\dnorm}{\hspace{-2pt}\downarrow} %

\newcommand{\drule}{\hspace{-.5pt}\downarrow}
\newcommand{\urule}{\hspace{-.5pt}\uparrow}

\newcommand{\T}{\mathcal{T}}
\newcommand{\A}{\mathcal{A}}
\renewcommand{\div}{\hspace{2pt}/}
\newcommand{\pfo}{\mathrm{PFO}}
\newcommand{\pushdown}{P\drule}
\newcommand{\pushup}{P\urule}

\newcommand{\splitdown}{S\drule}
\newcommand{\splitup}{S\urule}

\newcommand{\lrarr}{\leftrightarrow}
\newcommand{\lrarrstar}{\stackrel{*}{\lrarr}}

\newcommand{\rarr}{\rightarrow}
\newcommand{\rarrstar}{\stackrel{*}{\rarr}}
\newcommand{\larr}{\leftarrow}
\newcommand{\larrstar}{\stackrel{*}{\larr}}

\newcommand{\llb}{\llbracket}
\newcommand{\rrb}{\rrbracket}

\newcommand{\exwe}{\{ \exists, \wedge \}}
\newcommand{\fove}{\{ \forall, \vee \}}

\newcommand{\reals}{\mathbb{R}}

\renewcommand{\free}{\mathrm{free}}
\newcommand{\bagsize}{\mathsf{bagsize}}
\newcommand{\wid}{\mathsf{width}}

\section{Introduction}

The problem of evaluating a first-order sentence on a finite structure is
a central computational task in database theory, finite model theory, and indeed computer science at large.  In database theory, first-order sentences are viewed as constituting the core of SQL;
in parameterized complexity theory,
many commonly studied problems can be naturally and directly formulated
as cases of this evaluation problem.  In the quest to perform this evaluation efficiently,
the \emph{width} of a sentence---defined as the maximum number of free variables over all subformulas of the sentence---has become established as a central and crucial measure.
A first reason for this is that the natural bottom-up algorithm for sentence evaluation exhibits an exponential dependence
on the width~\cite{Williams14}; this algorithm runs in polynomial time when restricted to any class of sentences
having \emph{bounded width}, meaning that there exists a constant bounding their width.  In addition, there are a number of dichotomy theorems~\cite{GroheSS01,Grohe07,Chen14,ChenD16} showing that,
for particular fragments of first-order logic, a class of sentences admits tractable evaluation if and only if the class has bounded width, up to logical equivalence.
These theorems are shown under standard complexity-theoretic assumptions, and concern relational first-order logic, which we focus on in this article.

These results all point to and support the desirability of minimizing the widths of logical sentences.
However, an undecidability result~\cite[Theorem 19]{BovaChenICDT14} immediately rules out the existence of
an algorithm that, given a first-order sentence $\phi$, returns a sentence
of minimum width among all sentences logically equivalent to $\phi$; indeed,
this undecidability result rules out an algorithm that performs as described even on
positive first-order sentences.  
Despite this non-computability result, given the importance of
width minimization, it is natural to seek computable methods for 
reformulating sentences so as to reduce their width.
In this article, we show how to perform width minimization,
in an optimal fashion, via
established \emph{syntactic rewriting rules} known to preserve logical
equivalence.
The study of such rules is strongly relevant:
\begin{itemize}
\item In database theory, such rules are studied and
applied extensively to reformulate
database queries with the aim of allowing efficient evaluation~\cite[Chapter 6]{AbiteboulHV95}.

\item Well-known tractable cases of conjunctive query evaluation admit logical characterizations via such rules, such as those of bounded treewidth~\cite{cqc} and bounded hypertreewidth~\cite{robbers}.

\item Such rules play a key role in obtaining the tractability results of
some of the mentioned dichotomy theorems~\cite{Chen14,ChenD16}.
\end{itemize}
\noindent 
This article focuses on a collection of common
and well-known rewriting rules, which includes rules that
are well-established in database theory
\cite[Figures 5.1 and 6.2]{AbiteboulHV95}.
Our
main result is the presentation of an algorithm that,
given a positive first-order sentence $\phi$, outputs a minimum-width sentence
obtainable from $\phi$ via application of the considered rules.
We thus obtain a 
complete algorithmic understanding of width minimization up to the studied rules.
Our
main result is the first result---of which we are aware---to 
obtain a 
comprehensive characterization of width
up to syntactic rules in the general setting of positive first-order logic.

In order to obtain our main result, we use the theory of term rewriting; in particular,
we make use of basic concepts therein such as \emph{termination}
 and \emph{confluence}.
We view this marriage of \emph{term rewriting} and \emph{query rewriting}
as a conceptual contribution of our work.
We gently augment the basic theory of term rewriting
in two ways.
First, we define the notion of a \emph{gauged system},
which is an abstract reduction system along with a \emph{gauge}, a function mapping each
element of the system to a number; the gauge represents a quantity that one wants to minimize.
In this article our focus is on systems whose elements are positive first-order formulas,
and where the gauge is the width.
Second, we define a notion of division of a system by an equivalence relation.
This notion allows us to consider equivalence classes w.r.t.~a subset of the considered rules,
in a sense made precise.
In particular, we consider equivalence classes w.r.t.~a set of rules that
correspond to the computation of tree decompositions
(see Theorem~\ref{thm:width-for-regions}).
This will make precise and transparent the role of tree decomposition computations
in our algorithm; this also opens up the possibility of applying procedures
that approximate treewidth in a black-box fashion, or decomposition methods
that are designed for unbounded-arity query classes, 
such as those associated with
hypertree width~\cite{GottlobLS01}.
Relatedly, we mention here that for any Boolean conjunctive query, the minimum width 
obtainable by the studied rules is equal to the query's treewidth, plus one
(see Corollary~\ref{cor:exists-wedge-tw}); thus, this measure of minimum width,
when looked at for positive first-order logic, constitutes a generalization
of treewidth.

\emph{All together, the presented framework and theory initiate a novel interface among term rewriting, query rewriting, and structural decomposition theory.}  
In our view, the obtention of our main theorem necessitates the 
development of novel ideas and techniques that draw upon, advance,
and fuse together these areas.
We believe that further development of this interface has the potential to further unite these areas via asking new questions and prompting new techniques---of service to efficient query evaluation.
We wish to emphasize, as a conceptual contribution, the marriage
of \emph{term rewriting} and \emph{query rewriting},
which is (to our knowledge) novel, despite an enduring overlap in names!

We wish to add two technical remarks.
First, although our results are presented for positive first-order logic, our
width minimization procedure can straightforwardly be applied to general
first-order sentences by first turning it into negation normal form, i.e., pushing all negations down to the atomic level,
and then treating negated atoms as if they were atoms.
Second, our main result's algorithm straightforwardly gives rise to a
fixed-parameter tractability result: on any class of sentences that have bounded width
when mapped under our algorithm, we obtain fixed-parameter tractability of evaluating
the class, with the sentence as parameter, via invoking our algorithm and then performing the
natural bottom-up algorithm for sentence evaluation.

\paragraph*{Related work}
Some of the previous characterizations and studies of width, for example~\cite{Chen14,ChenD16,BovaC19}, focused on subfragments of positive first-order logic
defined by restricting the permitted quantifiers and connectives;
as mentioned, we here consider full positive first-order logic.
Width measures for first-order logic have previously been defined~\cite{AdlerWeyer,ChenJACM17}.
While the measure of \cite{ChenJACM17} makes use of syntactic rewriting rules,
and the work \cite{AdlerWeyer} studies the relationship of width to syntactic rewriting rules,
the present work shows how to optimally minimize formulas up to 
the set of studied syntactic rewriting rules, a form of result that
(to our knowledge) is not entailed by the theory of either of these
previous works.
Here, it is worth mentioning a separation that can be shown between the width measure of~\cite{AdlerWeyer} and the width measure introduced in this article:
the class of formulas in this article's Example~\ref{ex:adler} has unbounded width
under the width measure of \cite{AdlerWeyer} (see their Proposition 3.16), but has bounded width
  under the width notion of the present article.  Thus, insofar as the work~\cite{AdlerWeyer} exploits
  syntactic rewriting rules, it does not exploit the ones used in  Example~\ref{ex:adler}.

In our view, a positive feature of the present work 
is that our width measure has a
clear interpretation.
Moreover, we believe that the present work
also renders some of the previously defined
width measures as having, in retrospect, an ad hoc character:
as an example, the width measure of~\cite{ChenD16}
has the property that, when a formula has measure $w$,
it is possible to apply rewriting rules to the formula
so that it has width $w$, but the optimal width achievable
via these rules is not addressed; the considered rule set from~\cite{ChenD16} 
is encompassed by the rule set studied here.\footnote{
The rules used to rewrite formulas in~\cite{ChenD16} are presented
in the argument for~\cite[Theorem 5.1]{ChenD16}; 
to use the rule terminology of the present paper 
(given in Section~\ref{sect:rules}),
apart from associativity, commutativity, and ordering, the argument uses splitdown (for universal quantifications) and pushdown (for existential quantifications).
}

A conference version of this article appeared in the proceedings of ICDT '24. In contrast to that version, this one contains all proofs. Moreover, we have made many modifications to improve readability.

\section{Preliminaries}
\label{sct:prelims}

We assume basic familiarity with the syntax and semantics of relational first-order logic,
which is our logic of focus.
We assume relational structures to have nonempty universes.
In building formulas, we %
assume that conjunction ($\wedge$) and
disjunction ($\vee$) are binary.
When $\psi$ is a first-order formula, we use $\free(\psi)$ to denote the set of free variables of $\psi$.
By a \emph{pfo-formula}, we refer to a positive first-order formula, i.e., a first-order formula without any negation.
We use $\pfo$ to denote the class of all pfo-formulas.
The \emph{width} of a first-order formula $\phi$, denoted by $\wid(\phi)$, is defined as
the maximum of $|\free(\psi)|$ over all subformulas $\psi$ of $\phi$.
We will often identify pfo-formulas with their syntax trees, 
that is, trees where internal nodes are labeled with 
a quantification $Q v$ (where $Q$ is a quantifier and $v$ is a variable), a connective from $\{ \land, \lor \}$, and leaves are labeled with atoms.
We sometimes 
refer to a node of a syntax tree labeled by a connective $\oplus \in \{ \land, \lor \}$ simply as a \emph{$\oplus$-node}. 
We use the convention that, in syntax trees, quantifiers and the variables that they bind are in the label of a single node.

For a set $S$ of quantifiers and connectives, we define an $S$-formula to be a formula in which all quantifiers and connectives are from $S$; we analogously define $S$-sentences.

\subsection{Hypergraphs and tree decompositions}

A \emph{hypergraph} is a pair $(V,E)$
where $V$ is a set called the \emph{vertex set} of the hypergraph,  and $E$ is a subset of the
power set of $V$, that is, each element of $E$ is a subset of $V$; $E$ is called the
\emph{edge set} of the hypergraph.  Each element of $V$ is called a \emph{vertex},
and each element of $E$ is called an \emph{edge}. In this work, we only consider finite hypergraphs, so we tacitly assume that $V$ is always a finite set.
We sometimes specify a hypergraph simply via its edge set $E$; we do this with the understanding
that the hypergraph's vertex set is $\bigcup_{e \in E} e$.
When~$H$ is a hypergraph, we sometimes use $V(H)$ and $E(H)$ to denote its vertex set
and edge set, respectively.
We consider a \emph{graph} to be a hypergraph where every edge has size $2$.

A \emph{tree decomposition} of a hypergraph $H$ is a pair $(T, \{ B_t \}_{t \in V(T)})$
consisting of a tree $T$
and, for each vertex $t$ of $T$,
a \emph{bag} $B_t$ that is a subset of $V(H)$, such that the following conditions hold:
\begin{itemize}
\item (vertex coverage) for each vertex $v \in V(H)$, there exists a vertex $t \in V(T)$
  such that $v \in B_t$;
\item (edge coverage) for each edge $e \in E(H)$, there exists a vertex $t \in V(T)$
  such that $e \subseteq B_t$;
\item (connectivity) for each vertex $v \in V(H)$, the set $S_v = \{ t \in V(T) ~|~ v \in B_t \}$
  induces a connected subtree of $T$.
  
\end{itemize}
\noindent 
We define the \emph{bagsize} of a tree decomposition
$C = (T, \{ B_t \}_{t \in V(T)})$,
denoted by $\bagsize(C)$,
as $\max_{t \in V(T)} |B_t|$, that is, as the maximum size over all bags.
The \emph{treewidth} of a hypergraph $H$ is the minimum,
over all tree decompositions~$C$ of $H$,
of the quantity $\bagsize(C) - 1$.
A \emph{minimum width tree decomposition} of a hypergraph $H$ is a tree decomposition $C$
where $\bagsize(C) - 1$ is the treewidth of $H$. 

For every hypergraph $H$, there is a tree decomposition $C = (T, \{ B_t \}_{t \in V(T)})$ with $|V(T)| = O(\bagsize(H)|V(H)|)$. Moreover, given a tree decomposition $C' = (T', \{ B'_t \}_{t \in V(T')})$ of bagsize $k$ of $H$, one can compute a tree decomposition $C = (T, \{ B_t \}_{t \in V(T)})$ with $|V(T)| = O(\bagsize(H)|V(H)|)$ and bagsize $k$ in time $O(k^2 \cdot \max(V(T'), V(H)))$, see e.g.~\cite[Lemma~7.4]{CyganFKLMPPS15}. We will thus assume that all tree decompositions we deal with have size polynomial in the hypergraphs.

\subsection{Term rewriting}
\label{subsect:term-rewriting}
\paragraph*{Basics}
We here introduce the basic terminology of term rewriting
to be used; for the most part, we base our presentation on \cite[Chapter~2]{BaaderN98}.
A \emph{system} is a pair $(D, \rarr)$
where $D$ is a set, and $\rarr$ is a binary relation on $D$.
We refer to the elements of $D$ as elements of the system, and
the binary relation $\rarr$ will sometimes be called a \emph{reduction}.
Whenever $\rarr$ is a binary relation, we use $\larr$ to denote its \emph{inverse}
(so $a \rarr b$ if and only if $b \larr a$),
and we use $\lrarr$ to denote the union $\rarr \cup \larr$ which is straightforwardly verified to be a symmetric relation.
We use $\rarrstar$, $\larrstar$, and $\lrarrstar$ to denote
the reflexive-transitive closures of $\rarr$, $\larr$, and $\lrarr$,
respectively. Note that $\lrarrstar$ is an equivalence relation.

We will use the following properties of elements of a system.
Let $(D,\rarr)$ be a system, and let $d,e \in D$.
\begin{itemize}

\item $d$ is \emph{reducible} if there exists $d' \in D$ such that $d \rarr d'$.
\item $d$ is \emph{in normal form} if it is not reducible.
\item $e$ is a \emph{normal form} of $d$ if $d \rarrstar e$ and $e$ is in normal form.

  \item $d$ and $e$ are \emph{joinable}, denoted $d \downarrow e$, if there exists an element $f$ such that $d \rarrstar f \larrstar e$.
  
\end{itemize}
Let $\rarr'$ be a binary relation on a set $D$, and let $d \in D$; 
we say that $\rarr'$ is \emph{applicable to} $d$ or \emph{can be applied to} $d$ if there exists an element $e \in D$ such that $d \rarr' e$.
We extend this to sets of elements in a natural way: when $S \subseteq D$,
we say that $\rarr'$ is \emph{applicable to} $S$ or \emph{can be applied to} $S$ if there exist elements $d \in S$ and  $e \in D$ such that $d \rarr' e$.

We next define properties of a system $Y = (D,\rarr)$; in what follows,
$d$, $e$, $e'$, $d_i$, etc. denote elements of $D$.
\begin{itemize}
\item $Y$ is \emph{confluent} if $e \larrstar d \rarrstar e'$ implies $e \downarrow e'$.
\item $Y$ is \emph{locally confluent} if $e \larr d \rarr e'$ implies $e \downarrow e'$.
\item $Y$ is \emph{terminating} if there is no infinite descending chain
  $d_0 \rarr d_1 \rarr \cdots$ of elements in $D$.
\item $Y$ is \emph{convergent} if it is confluent and terminating.
\end{itemize}
We remark that a confluent system is also locally confluent, 
since with respect to a system, $d \rarr e$ implies $d \rarrstar e$.
Also note that the property of being terminating immediately implies 
that of being \emph{normalizing}, 
meaning that each element has a normal form. 

We will use the following properties of convergent systems.

\begin{proposition}[see~\cite{BaaderN98}, Lemma 2.1.8 and Theorem 2.1.9]\label{prop:samenormalform} 
  Suppose that $(D,\rarr)$ is a convergent system.
  Then, each element $d \in D$ has a unique normal form.
  Moreover, for all elements $d, e \in D$, it holds that
  $d \lrarrstar e$ if and only if they have the same normal form.
\end{proposition}

\emph{Whenever an element $d\in D$ has a unique normal form, we denote this unique normal form by $d \dnorm$.}
The last part of Proposition~\ref{prop:samenormalform} can then be formulated as follows: for every convergent system $(D,\rarr)$ we have 
for all elements $d, e \in D$ that $d \lrarrstar e$ if and only if $d \dnorm = e \dnorm$.

The following lemma, often called \emph{Newman's Lemma}, was first shown in~\cite{Newman42} and will be used to establish convergence. See also~\cite[Lemma 2.7.2]{BaaderN98} for a proof.

\begin{lemC}[\cite{Newman42}]
  \label{lemma:newman}
  If a system is locally confluent and terminating, then it is confluent, and hence convergent.
\end{lemC}

\paragraph*{Gauged systems}
We next extend the usual setting of term rewriting by
considering systems having a~\emph{gauge}, 
which intuitively is a measure that one desires to minimize by rewriting.
To this end, define a \emph{gauged system} to be a triple $(D,\rarr,g)$ where
$(D,\rarr)$ is a system, and $g: D \to \reals$ is a mapping called a \emph{gauge}.
A gauged system is \emph{monotone} if for any pair of elements $d,e \in D$
where $d \rarr e$, it holds that $g(d) \geq g(e)$, that is, applying the reduction
cannot increase the gauge.
We apply the terminology of
Section~\ref{subsect:term-rewriting} to gauged systems;
for example, we say that a gauged system $(D,\rarr,g)$ is
\emph{convergent} if the system $(D,\rarr)$ is convergent.

\begin{propositionrep}
  \label{prop:monotone-convergent}
  Let $(D,\rarr,g)$ be a gauged system that is monotone and convergent.
  Then, for any elements $d, e \in D$, it holds that $d \lrarrstar e$
  implies $g(d \dnorm) \leq g(e)$.  
That is, for any element $d$, its normal form
  $d \dnorm$ has the minimum gauge among all elements $e$ with $d \lrarrstar e$.
\end{propositionrep}

\begin{proof}
First note that, by a simple induction using the monotonicity of $g$, we have for all elements $e', e''\in D$ that $e' \rarrstar e''$ implies $g(e')\ge g(e'')$.
Next, since $(D, \rarr)$ is convergent, we have by Proposition~\ref{prop:samenormalform} that $d \dnorm = e \dnorm$. 
But then $g(e) \ge g(e\dnorm) = g(d\dnorm)$.
\end{proof}

\paragraph*{Systems and division}
We here define a notion of dividing systems by equivalence relations.
Let $D$ be a set and let~$\equiv$ be an equivalence relation on $D$.
Let $D / \hspace{-1pt} \equiv$ denote the set containing the $\equiv$-equivalence classes.
We define $(D,\rarr) / \hspace{-1pt} \equiv$ as the system $(D / \hspace{-1pt} \equiv, \rarr)$,
where we extend the definition of $\rarr$ to $D / \hspace{-1pt} \equiv$ by positing that
for $\equiv$-equivalence classes $C, C'$, it holds that
$C \rarr C'$ if and only if there exist $d \in C$ and $d' \in C'$ such that
$d \rarr d'$.

Suppose that $G = (D,\rarr,g)$ is a gauged system.
We extend the definition of $g$ so that,
for each set $S \subseteq D$, we have $g(S) = \inf \{ g(d) ~|~ d \in S \}$.
We define
$G / \hspace{-1pt} \equiv$ as the gauged system $(D / \hspace{-1pt} \equiv, \rarr, g)$.

\section{Rewriting rules}
\label{sect:rules}

We here define well-known transformation rules on first-order formulas,
which are all known to preserve logical equivalence.  We define these rules
with the understanding that they can always be applied to subformulas,
so that when $\psi_1$ is a subformula of $\phi_1$ and $\psi_1 \rarr \psi_2$,
we have $\phi_1 \rarr \phi_2$, where~$\phi_2$ is obtained from $\phi_1$
by replacing the subformula $\psi_1$ with~$\psi_2$.
In the following $F_1, F_2$, etc. denote formulas. It will be convenient to formalize these rules as rewriting rules in the sense defined before.

The rule $\rarr_A$ is \emph{associativity} for $\oplus \in \{\land, \lor\}$:%
\begin{align*}
	(F_1\oplus (F_2 \oplus F_3))\rarr_{A} ((F_1\oplus F_2)\oplus F_3), \quad %
((F_1\oplus F_2)\oplus F_3)\rarr_{A} (F_1\oplus (F_2 \oplus F_3)).
  \end{align*}

The rule $\rarr_C$ is \emph{commutativity} for $\oplus \in \{\land, \lor\}$:%
  \begin{align*}
    (F_1 \oplus F_2)\rarr_{C} (F_2\oplus F_1).
  \end{align*}

The \emph{reordering} rule $\rarr_{O}$ is defined when $Q \in \{ \forall, \exists \}$
is a quantifier: $Q x Q y F \rarr_O Q y Q x F$.

  The \emph{pushdown} rule $\rarr_{\pushdown}$ applies when $F_2$ is a formula in which $x$ is not free:
  \begin{align*}
    \exists x (F_1 \wedge F_2) \rarr_{\pushdown} (\exists x F_1) \wedge F_2,\quad %
    \forall x (F_1 \vee F_2) \rarr_{\pushdown} (\forall x F_1) \vee F_2.
  \end{align*}
The \emph{pushup} rule $\rarr_{\pushup}$ is defined as the inverse of $\rarr_{\pushdown}$.

The \emph{renaming} rule $\rarr_{N}$ allows $Q x F \rarr_N Q y F'$
when $y \notin \free(F)$ and $F'$ is derived from $F$ by replacing each free occurrence of $x$ with $y$.

Collectively, 
we call the above rules the \emph{tree decomposition rules}, since, as we will see in Section~\ref{sect:tds}, they allow, in a way that we will make precise, optimizing the treewidth of certain subformulas.

The \emph{splitdown} rule distributes quantification over a corresponding connective:
\begin{align*}
  \exists x (F_1 \vee F_2) \rarr_{\splitdown} (\exists x F_1) \vee (\exists x F_2), \quad %
  \forall x (F_1 \wedge F_2) \rarr_{\splitdown} (\forall x F_1) \wedge (\forall x F_2).
\end{align*}
The \emph{splitup} rule $\rarr_{\splitup}$ is defined as the inverse of $\rarr_{\splitdown}$.

The \emph{removal} rule $\rarr_M$  allows removal of a quantification when no variables are bound to the quantification: when $Q \in \{ \forall, \exists \}$, we have
$Q x F \rarr_M F$ when $F$ contains no free occurrences of the variable $x$,
that is, when $x \notin \free(F)$.  (Note that we assume that formulas under discussion are evaluated over structures with non-empty universe; this assumption is needed for this removal rule to preserve logical equivalence.)

We tacitly use the straightforwardly verified fact that when the given rules are applied to any pfo-formula,
the formula's set of free variables is preserved.
We denote these rules via the given subscripts.
We use $\T$ to denote the set of all tree decomposition rules,
so $\T = \{ A, C, \pushdown, \pushup, O, N \}$.
We use $\A$ to denote the set of all presented rules, so
$\A = \T \cup \{ \splitdown, \splitup, M \}$.

When we have a set $\mathcal{R}$ of subscripts representing these rules, we use
$\rarr_{\mathcal{R}}$ to denote the union $\bigcup_{R \in \mathcal{R}} \rarr_R$.
For example, $\rarr_{\{ A,C,O \}}$ denotes $\rarr_{A} \cup \rarr_{C} \cup \rarr_{O}$.
In this context, we sometimes denote a set by the string concatenation of its elements,
for example, we write $\rarr_{ACO}$ in place of $\rarr_{\{ A,C,O \}}$.
We apply these conventions to $\larr, \lrarr, \rarrstar, \larrstar, \lrarrstar$, and so forth.
Thus $\lrarr_{ACO}$ denotes $\lrarr_{\{ A,C,O \}}$, which
is equal to 
$\rarr_{A} \cup \rarr_{C} \cup \rarr_{O} \cup \larr_{A} \cup \larr_{C} \cup \larr_{O}$.
When $\mathcal{R}$ is a set of rule subscripts
and $\phi$ is a formula, we use $[\phi]_{\mathcal{R}}$
to denote the $\lrarrstar_{\mathcal{R}}$-equivalence class of $\phi$.
So, for example, $[\phi]_{\T}$ denotes the $\lrarrstar_{\T}$-equivalence class of $\phi$,
and $[\phi]_{ACO}$ denotes the $\lrarrstar_{ACO}$-equivalence class of $\phi$.

\begin{figure}
\begin{center}
\begin{tabular}{ccc}
  Shorthand & Rule name & Brief description \\
  \toprule
  $A$  & associativity & associativity of connectives $\land, \lor$  \\
  $C$  & commutativity  & commutativity of connectives $\land, \lor$  \\
  $O$ & reordering & $Qx Qy F \rarr Qy Qx F$ when $Q \in \{ \exists, \forall \}$ \\
  $\pushdown$ & pushdown &     $\exists x (F_1 \wedge F_2) \rarr (\exists x F_1) \wedge F_2$ when $x \notin \free(F_2)$, and dual \\
  $\pushup$ & pushup & inverse of pushdown \\
  $N$ & renaming & renaming $x$ as $y$ in $Qx F$ when $y \notin \free(F)$ \\
\midrule
  $\splitdown$ & splitdown &  $\exists x (F_1 \vee F_2) \rarr (\exists x F_1) \vee (\exists x F_2)$, and dual \\
  $\splitup$ & splitup & inverse of splitdown \\
  $M$ & removal & $Q x F \rarr F$ when $x \notin \free(F)$ \\
  \bottomrule
\end{tabular}

\end{center}
\caption{Summary of studied rules.}
\label{fig:rules}
\end{figure}

\begin{example}
\label{ex:role-of-O}
Let
 $\phi = \exists x \exists y \exists t (R(x,t) \wedge S(t,y))$; this sentence has width $3$.
We have 
$$\phi \thickspace \rarrstar_{O} \thickspace \exists t \exists x \exists y (R(x,t) \wedge S(t,y)) \thickspace \rarrstar_{C,\pushdown} \thickspace \exists t ( (\exists x R(x,t)) \wedge (\exists y S(t,y)) ).$$
 The last formula has width $2$; the rules we used were %
in~$\T$.
It can, however, be verified that,
using the rules in $\T \setminus \{ O \}$,
this formula $\phi$ cannot be rewritten into a width $2$ formula:
using these rules, any rewriting will be derivable from $\phi$
via commutativity and renaming.
\end{example}

\begin{example}
  \label{ex:adler}
  Consider the formulas $(\phi_n)_{n \geq 1}$ defined by
  $\phi_n = \exists x_1 \ldots \exists x_n \forall y (\bigwedge_{i=1}^n E_i(x_i,y))$; the formula $\phi_n$ has width $n+1$.
  (Although formally we consider conjunction as a binary connective, 
we allow higher-arity
  conjunctions due to the presence of the associativity rule.)
  We have 
$$\phi_n \thickspace \rarrstar_{\splitdown,A}  \thickspace \exists x_1 \ldots \exists x_n (\bigwedge_{i=1}^n \forall y E_i(x_i,y))  \thickspace \rarrstar_{\pushdown,A,C}  \thickspace \bigwedge_{i=1}^n (\exists x_i \forall y E_i(x_i,y)).$$
  The last formula has width $2$, and hence the rules we consider suffice to convert each formula $\phi_n$
  into a width $2$ formula.
\end{example}

\paragraph*{Justification of the rule choice}

Before stating our main result in the next section, let us take some time to discuss why we have chosen the above rules for study, in this paper. First, as suggested in the article, these rules appear in a number of textbooks and well-known sources. In addition to being in the standard database book~\cite{AbiteboulHV95}, one can find some of the crucial ones in~\cite[page 99]{Papadimitriou94}, and many of the crucial ones also appear in the Wikipedia entry on first-order logic.\footnote{\url{https://en.wikipedia.org/wiki/First-order_logic\#Provable_identities}} We can remark that these rules are generally used when showing that first-order formulas can be rewritten into prenex normal form.

Indeed, apart from distributivity (discussed in the conclusion), we are not aware of any other syntactic rewriting rules (apart from combinations of the rules that we study); in particular, we did not find any others in any of the standard sources that we looked at.

As alluded to in the introduction, the study of many of these rules is strongly established in the research literature. They are studied in articles including \cite{AdlerWeyer,BovaChenICDT14,ChenD16,Chen14,ChenJACM17,DalmauKV02,robbers,cqc}.
As mentioned, subsets of the rule set are used crucially, in the literature, to obtain the positive results of dichotomy theorems. In particular, there are precise contexts where a subset of the rule set is sufficient to give a tractability result (for example, in~\cite{ChenD16}). The tractability result of the present article strengthens these previous tractability results since it optimally minimizes width with respect to the considered rules, and thus provides a wider and deeper perspective on these previous tractability results.

Let us emphasize that a subset of the considered rules corresponds to treewidth computation in a very precise sense (see Theorem~\ref{thm:width-for-regions} and Corollary~\ref{cor:exists-wedge-tw}). The correspondence between rewriting rules and computation of tree decompositions was also studied, for example, in~\cite{cqc, DalmauKV02,robbers}.

\section{Main theorem and roadmap}

In this section, 
we first state our main theorem, then we give the intermediate results that we will show in the remainder of the paper to
derive it; at the end of the section, we prove %
the main theorem using the intermediate results.
Our main theorem essentially says the following: 
there is an algorithm that,
given a pfo-formula $\phi$, computes a formula $\phi'$ 
that is derived from $\phi$
by applying the defined rules $\A$ and that has the minimum width
over formulas derivable by these rules.
Moreover, up to computation of minimum width tree decompositions, the algorithm computing $\phi'$ is efficient in that it runs in polynomial time.

\begin{theorem}
\label{thm:main}
  (Main theorem)
There exists an algorithm $A$ that, given a pfo-formula $\phi$,
computes a minimum width element of $[\phi]_{\A}$.
Moreover, with oracle access to an algorithm for computing minimum width tree decompositions,
the algorithm $A$ can be implemented in polynomial time.
\end{theorem}
\noindent 
Here, we say that an algorithm \emph{computes minimum width tree decompositions} if, when given a hypergraph $H$,
it outputs a minimum width tree decomposition of $H$.

In order to establish our main theorem, 
we study and make use of the systems
$$	Y = (\pfo, \to_{\{\pushdown,\splitdown,M\}}) \div \lrarrstar_{ACO},\quad
	Y' = (\pfo, \to_{\{\splitdown,M\}}) \div \lrarrstar_{\T}.$$
We also make use of the gauged system
$$G' = (\pfo, \to_{\{\splitdown,M\}}, \wid) \div \lrarrstar_{\T},$$ where, as defined before, $\wid$ is the function that 
maps a formula to its width.
We show that for the system $Y$, normal forms can be computed in polynomial time,
and moreover, that a normal form of $Y$ directly yields a normal form of $Y'$.

\begin{theorem}
  \label{thm:push-in}
The system $Y$ %
is terminating; moreover,
there exists a polynomial-time algorithm that, given a pfo-formula $\phi$,
returns a pfo-formula $\phi^+$ such that
$[\phi^+]_{ACO}$ is a normal form of $[\phi]_{ACO}$ in the system $Y$.
\end{theorem}

\begin{restatable}{theorem}{theoremnormalforms}
  \label{thm:normal-forms}
  Suppose that $\phi$ is a pfo-formula where
  $[\phi]_{ACO}$ is a normal form of the system~$Y$;
then, $[\phi]_{\T}$ is a normal form of the system
$Y'$.
\end{restatable}
\noindent 
We then prove that the system $Y'$ is convergent, and that its corresponding
gauged system $G'$ is monotone.  Together, these results allow us to leverage
Proposition~\ref{prop:monotone-convergent} and allow us to compute
minimum-width equivalence classes in $G'$.

\begin{theorem}
  \label{thm:convergent}
The system $Y'$ is convergent.
\end{theorem}

\begin{restatable}{theorem}{theoremmonotone}
    \label{thm:monotone}
The gauged system $G'$ is monotone.
\end{restatable}
\noindent 
From the previous results, we are able to present an algorithm that,
given a pfo-formula $\phi$, yields a pfo-formula $\phi^+$ 
obtainable from $\phi$ via the studied rewriting rules,
where
$[\phi^+]_{\T}$ is a normal form of the system $Y'$ and hence of $G'$.
In order to derive the main theorem, 
the remaining piece needed is to show that minimization can be performed within a $\lrarrstar_{\T}$-equivalence class---that is, that from $\phi^+$,
a minimum width element from the set $[\phi^+]_{\T}$ can be computed.
This piece is supplied by the following theorem.

\begin{restatable}{theorem}{theoremminimizeinT}
\label{thm:minimize-in-T}
  There exists an algorithm that, given a pfo-formula $\theta$,
  outputs a pfo-formula $\theta^+$ having minimum width among all formulas in
  $[\theta]_{\T}$.  With oracle access to an algorithm for computing minimum width tree decompositions,
  this algorithm can be implemented in polynomial time.
\end{restatable}
\noindent 
Before giving the formal proof of the main theorem, let us describe briefly how and why its algorithm works.  Given a pfo-formula $\phi$, the main theorem's algorithm first invokes the algorithm of Theorem~\ref{thm:push-in}, which yields a pfo-formula $\theta$ where, in the system $Y$, 
it holds that $[\theta]_{ACO}$ is a normal form of $[\phi]_{ACO}$.
By Theorem~\ref{thm:normal-forms} and some extra reasoning,
we have that, in the system $Y'$,
$[\theta]_{\T}$ is a normal form of $[\phi]_{\T}$.
The algorithm then uses as a subroutine
the algorithm of Theorem~\ref{thm:minimize-in-T} to
compute and output a minimum width element of 
 $[\phi]_{\T}$.
The correctness of this output is provided by 
Theorems~\ref{thm:convergent} and~\ref{thm:monotone} in
conjunction with Proposition~\ref{prop:monotone-convergent};
these two theorems provide that the system $Y'$ and its respective gauged system
$G'$ are convergent and monotone, respectively.

We next give an example explaining how the main theorem's algorithm works.

\begin{exa}
Consider the example formula 
\[\phi_0 = \forall y \forall z \exists v_1 \exists v_3 \exists v_2 (T(y,z) \vee (E(v_1,v_2) \wedge (E(v_2,v_3) \wedge E(v_3,z)))).\]
By applying each of the splitdown and removal rules once to the quantification
$\exists v_2$, we obtain  $\phi_0 \rarrstar_{\splitdown,M} \phi_1$
where 
\[\phi_1 = \forall y \forall z \exists v_1 \exists v_3 (T(y,z) \vee \exists v_2 (E(v_1,v_2) \wedge (E(v_2,v_3) \wedge E(v_3,z)))).\]
Now observe that, in the formula $\phi_1$,
 we cannot apply the pushdown rule to the quantification
$\exists v_2$, since when viewing the formula as a tree,
this quantification is on top of a conjunction where $v_2$ 
appears as a free variable on both sides of the conjunction.  
However, if we apply associativity,
we can obtain the formula 
\[ \phi'_1 = \forall y \forall z \exists v_1 \exists v_3 (T(y,z) \vee \exists v_2 ((E(v_1,v_2) \wedge E(v_2,v_3)) \wedge E(v_3,z))).
\]
where $\phi_1 \rarr_A \phi'_1$.  Since $\phi'_1$ can be derived from $\phi_1$ 
via applications of the ACO rules, we have $[\phi_1]_{ACO} = [\phi'_1]_{ACO}$.
Then, we can apply the pushdown rule to the quantification $\exists v_2$
to obtain $\phi'_1 \rarr_{\pushdown} \phi_2$ where
\[ \phi_2 = \forall y \forall z \exists v_1 \exists v_3 (T(y,z) \vee (\exists v_2 (E(v_1,v_2) \wedge E(v_2,v_3))) \wedge E(v_3,z)).
\]

We have $[\phi_0]_{ACO}  \rarrstar_{\splitdown,M} [\phi_1]_{ACO} = [\phi'_1]_{ACO} \rarr_{\pushdown} [\phi_2]_{ACO}$,
which immediately implies
$[\phi_0]_{ACO} \rarrstar_{\pushdown,\splitdown,M} [\phi_2]_{ACO}$.

By applying splitdown and removal to the quantification $\exists v_3$
and then to the quantification $\exists v_1$, we obtain 
\[ \phi_3 = \forall y \forall z (T(y,z) \vee  \exists v_1 \exists v_3 ((\exists v_2 (E(v_1,v_2) \wedge E(v_2,v_3))) \wedge E(v_3,z)))
\]
with $\phi_2 \rarrstar_{\splitdown,M} \phi_3$.

The pushdown rule cannot be applied in $\phi_3$: the only quantifications
that are above a connective are $\forall z$, $\exists v_3$, and $\exists v_2$,
and in none of these cases can the pushdown rule be applied.
However, by applying the reordering rule to $\forall y \forall z$
and to $\exists v_1 \exists v_3$
we obtain from $\phi_3$ the formula
\[ \phi'_3 = \forall z \forall y (T(y,z) \vee  \exists v_3 \exists v_1 ((\exists v_2 (E(v_1,v_2) \wedge E(v_2,v_3))) \wedge E(v_3,z)))
\]
where in each of these pairs of quantifications, the quantifications have been swapped;
we have $\phi_3 \rarrstar_O \phi'_3$.
Then, the pushdown rule can be applied to $\phi'_3$ at both of the quantifications $\forall y$ and $\exists v_1$; we have
$\phi'_3 \rarrstar_{\pushdown} \phi_4$, where
\[ \phi_4 = \forall z ((\forall y T(y,z)) \vee  \exists v_3 ((\exists v_1 \exists v_2 (E(v_1,v_2) \wedge E(v_2,v_3))) \wedge E(v_3,z))).
\]
By applying the reordering rule to $\exists v_1 \exists v_2$ and then
the pushdown rule to $\exists v_1$, we obtain
\[ \phi_5 = \forall z ((\forall y T(y,z)) \vee  \exists v_3 ((\exists v_2  (\exists v_1 E(v_1,v_2) \wedge E(v_2,v_3))) \wedge E(v_3,z))).
\]
We have  $[\phi_2]_{ACO} \rarrstar_{\splitdown,M} [\phi_3]_{ACO} = [\phi'_3]_{ACO} \rarr_{\pushdown} [\phi_4]_{ACO}$.
From the way that we derived~$\phi_5$ from $\phi_4$,
we have that there exists a formula $\phi'_4 \in  [\phi_4]_{ACO}$
such that $\phi'_4 \rarr_{\pushdown} \phi_5$; thus
$[\phi_4]_{ACO} \rarr_{\pushdown} [\phi_5]_{ACO}$.

We have derived the relationships
\[ [\phi_0]_{ACO} \rarrstar_{\pushdown,\splitdown,M} [\phi_2]_{ACO}  \rarrstar_{\pushdown,\splitdown,M} [\phi_4]_{ACO}  \rarrstar_{\pushdown,\splitdown,M} [\phi_5]_{ACO},\]
implying that 
$[\phi_0]_{ACO} \rarrstar_{\pushdown,\splitdown,M} [\phi_5]_{ACO}$.
In $\phi_5$, the last instance of $\wedge$ that appears has as its conjuncts
the formulas $(\exists v_2  (\exists v_1 E(v_1,v_2) \wedge E(v_2,v_3)))$ and 
$E(v_3,z)$.  Thus in looking at the syntax tree of $\phi_5$,
there are no two instances of the same connective adjacent to each other,
and so the associativity rule cannot be applied to $\phi_5$.
In this syntax tree, it can also be verified that
there are no two instances of quantification adjacent to each other,
and so the reordering rule cannot be applied to $\phi_5$.
Thus, each element of the equivalence class 
$ [\phi_5]_{ACO}$ can be obtained from $\phi_5$ solely by applying
the commutativity rule.  It can be verified readily that
among the rules pushdown, splitdown, and removal, none of them
can be applied to any formula in 
$ [\phi_5]_{ACO}$.  Thus, the reduction of the system $Y$
cannot be applied to $[\phi_5]_{ACO}$, which is an element of the system $Y$.
Hence, we have that $[\phi_5]_{ACO}$ is a normal form of the system $Y$.

Our main theorem's algorithm, having found this
formula $\phi_5$ where $[\phi_5]_{ACO}$ is a normal form in the system $Y$,
then uses the algorithm of
 Theorem~\ref{thm:minimize-in-T} to
find a minimum width element of 
$[\phi_5]_{\T}$.
Let us put aside the renaming rule (rule $N$), which
one can verify from Proposition~\ref{prop:getridofN} is actually
no longer needed to minimize width at this point, since $\phi_5$
and indeed all elements in $[\phi_5]_{\T}$ are \emph{standardized} 
(see Section~\ref{subsect:standardized} for the definition of this notion; note that Theorem~\ref{thm:minimize-in-T}'s algorithm standardizes the formula that it is given upfront, and then also subsequently puts aside the renaming rule).
It follows from a result to be shown 
(namely, Proposition~\ref{prop:T-to-organized})
that the subformula 
\[\psi =  \exists v_3 ((\exists v_2  (\exists v_1 E(v_1,v_2) \wedge E(v_2,v_3))) \wedge E(v_3,z))\]
of $\phi_5$ will remain intact up to applying rules in $\T \setminus \{N\}$ to
$\phi_5$, in that any formula reachable from~$\phi_5$ by applying
rules in $\T \setminus \{N\}$ will contain a subformula reachable from $\psi$
by applying rules in $\T \setminus \{N\}$, that is, a subformula contained in $[\psi]_{\T \setminus \{N\}}$.
Part of what Theorem~\ref{thm:minimize-in-T}'s algorithm achieves
is to find a minimum width element of $[\psi]_{\T \setminus \{N\}}$
using tree decomposition methods.  In this particular situation,
it can be verified that $\psi$ itself is a 
 minimum width element of~$[\psi]_{\T \setminus \{N\}}$: the formula $\psi$ has width~$2$,
and any element of~$[\psi]_{\T \setminus \{N\}}$ must have width $2$,
since $\psi$ has atoms having width~$2$, and each atom present
in $\psi$ must be present in each element of $[\psi]_{\T \setminus \{N\}}$.
\end{exa}

\begin{proof}[Proof of the main theorem---Theorem~\ref{thm:main}]
Given a pfo-for\-mula $\phi$, the algorithm first applies the algorithm of
Theorem~\ref{thm:push-in} to obtain a pfo-formula $\theta$ where
$[\theta]_{ACO}$ is a normal form of $[\phi]_{ACO}$ in $Y$,
and then applies the algorithm of Theorem~\ref{thm:minimize-in-T}
to obtain a pfo-formula $\theta^+$ having minimum width among the formulas
in $[\theta]_{\T}$; $\theta^+$ is the output of the algorithm.

We justify this algorithm's correctness as follows.
As $[\theta]_{ACO}$ is a normal form of $[\phi]_{ACO}$ in $Y$,
we have $[\phi]_{ACO} \rarrstar_{\{\pushdown,\splitdown,M\}} [\theta]_{ACO}$.
It follows that
 $[\phi]_{\T} \rarrstar_{\{\splitdown,M\}} [\theta]_{\T}$;
this is due to the observations that
(1) $[\psi]_{ACO} \rarr_{\splitdown, M} [\psi']_{ACO}$ implies
 $[\psi]_{\T} \rarr_{\splitdown, M} [\psi']_{\T}$ 
(as $\{A, C, O\} \subseteq \T$), and
(2) $[\psi]_{ACO} \rarr_{\pushdown} [\psi']_{ACO}$ implies
 $[\psi]_{\T} = [\psi']_{\T}$.
It follows from Theorem~\ref{thm:normal-forms} that $[\theta]_{\T}$ is a normal form
of $[\phi]_{\T}$ in $Y'$.
From Theorems~\ref{thm:convergent} and~\ref{thm:monotone}, we have that
the gauged system $G'$ is convergent and monotone, so by
Proposition~\ref{prop:monotone-convergent}, we have that,
in the gauged system $G'$,
the element $[\theta]_{\T}$ has the minimum gauge among all elements (of $G'$)
that are $\lrarrstar_{\{\splitdown,M\}}$-related to $[\phi]_{\T}$,
or equivalently, among all elements (of $G'$) that are contained in 
$[\phi]_{\T \cup \{ \splitdown, \splitup, M \} } $.
Thus, a minimum width element in $[\theta]_{\T}$ is a minimum width element
in $[\phi]_{\T \cup \{ \splitdown,\splitup,M \} }$.
\end{proof}

\section{Formulas}\label{sct:formulas}

In this section, we define a few types of formulas to be used
in our development, and show some basic properties thereof.

\subsection{Standardized formulas}
\label{subsect:standardized}
We first present the notion of \emph{standardized formula};
intuitively, a standardized formula is one where there is no name clash between
a quantified variable $x$ and other variable occurrences in the formula.
We will show a proposition implying that, 
given any formula $\phi$,
this formula can be converted into a standardized formula $\phi'_0$ via
applications of the renaming rule, and moreover, if one has the goal
of minimizing width via the presented rules, one can simply 
minimize the width of $\phi'_0$ via the presented rules minus the renaming
rule; that is, to minimize the width of $\phi'_0$, one never needs
the renaming rule.
Define a \emph{standardized} formula to be a formula $\phi$ where, for each occurrence $Q x$
of quantification, the following hold: (1) $x$ is not quantified elsewhere, that is,
for any other occurrence $Q' x'$ of quantification, $x \neq x'$ holds;
(2) $x \notin \free(\phi)$.
It is straightforward to verify that a subformula of a standardized formula is also standardized.

We will use the following proposition which 
implies that, in terms of applying the studied rules to reduce width,
one may always assume that variable renaming that leads to a standardized formula
is always performed upfront, and afterwards, no variable renaming is used.

\begin{propositionrep}\label{prop:getridofN}
  Let $\phi$ be a pfo-formula.
  Suppose $\phi = \phi_0$ and $\phi_0 \rarr_{A_1} \phi_1 \cdots \rarr_{A_t} \phi_t$
  where each $A_i$ is in $\T  \cup \{\splitdown, M \}$.
  Let $A'_1, \ldots, A'_s$ be the sequence obtained from $A_1, \ldots, A_t$
  by removing instances of $N$. Then:
\begin{enumerate}
\item There exists a standardized formula $\phi'_0$ with the following property: 

 there exist
 formulas $\phi'_1, \ldots, \phi'_s$
  where 
$\phi \rarrstar_{N} \phi'_0 \rarr_{A'_1} \phi'_1 \cdots \rarr_{A'_s} \phi'_s \rarrstar_{N} \phi_t$.  

\item For any standardized formula $\phi'_0$ with
$\phi \rarrstar_{N} \phi'_0$, the  property given in (1) holds.
\end{enumerate}
\end{propositionrep}
\begin{proof}
    We start by describing the first renaming sequence in the proposition. 
We view formulas as syntax trees.  For every node $v$ in $\phi$ labeled with a quantification $Q x$, introduce a variable $x_v$ that does not appear anywhere else in $\phi$ and is different for every $v$. Then the sequence of $\rarr_N$ applications renames for every quantification $Q x$ the variable to~$x_v$. Let $\phi_0'$ be the resulting formula. By definition $\phi \rarrstar_N \phi_0'$ and~$\phi_0'$ is standardized, since all quantifiers now are on a newly introduced variable not appearing in $\phi$ which are thus not free in $\phi_0'$ and appear in no other quantification. 
    The remainder of the proof assumes that~$\phi_0'$ is standardized,
    and shows that the given property holds, yielding the entire proposition.
    
    We now apply the sequence of rules $A_1', \ldots, A_s'$ on $\phi_0'$, by 
using the locations where the respective rules were applied to $\phi$.
Note that this is possible since all rules we consider work on at most a single quantifier and are independent of variable names. Let the result be $\phi_s'$.
    
    It remains to rename the variables of $\phi_s'$. To this end, observe that by a simple induction along the reduction sequence we can see that there is an isomorphism $b$ between~$\phi_s'$ and $\phi_t$ (which was the last formula in the given sequence). In particular, there is for every node $v$ in $\phi_s'$ labeled by a quantifier the corresponding node $b(v)$ in $\phi_t$ that is also labeled by the same quantifier. So we simply rename, for every such node $v$, the variable quantified at $v$ to the variable quantified at $b(v)$; this allows us to rewrite $\phi_s'$ to $\phi_t$ by only applying $\rarr_N$.
\end{proof}

Note that $ \phi'_s \rarrstar_{N} \phi_t$ implies that
$\phi'_s$ and $\phi_t$ have the same width, since applying the renaming rule does not change the width.  Hence if one wants to optimally minimize width up to all of the presented rules, this proposition yields that one can first standardize a given formula, and then minimize with respect to all of the presented rules minus the renaming rule.

\begin{example}
Consider the formula $\phi =  (\exists w (E(u,w) \land E(w,v))) \land (\exists u (E(v,u) \land E(u,w)))$.  
This formula $\phi$ is not standardized, since there exists an instance of quantification $Qw$ whose variable is free in $\phi$: we have $w \in \free(\phi)$ since $w$ is a free variable of the conjunct 
$(\exists u (E(v,u) \land E(u,w)))$, and hence of $\phi$ itself.
We have 
\begin{eqnarray*}
\phi & \rarr_N & (\exists c (E(u,c) \land E(c,v))) \land (\exists u (E(v,u) \land E(u,w))) \\
& \rarr_{\pushup} & \exists c \hspace{3pt} ( (E(u,c) \land E(c,v)) \land (\exists u (E(v,u) \land E(u,w)))) \\
& \rarr_{C} &
\exists c \hspace{3pt} ( 
(\exists u (E(v,u) \land E(u,w)))
\land 
(E(u,c) \land E(c,v)) 
) \\
& \rarr_N & 
\exists c \hspace{3pt} ( 
(\exists d (E(v,d) \land E(d,w)))
\land 
(E(u,c) \land E(c,v))  \\
& \rarr_{\pushup} &
\exists c \exists d \hspace{3pt} ( 
(
(E(v,d) \land E(d,w))
\land 
(E(u,c) \land E(c,v))
).
\end{eqnarray*}
\noindent 
We view this example sequence as the sequence
 $\phi = \phi_0 \rarr_{A_1} \phi_1 \cdots \rarr_{A_5} \phi_5$
containing~$5$ rule applications:
$A_1 = N, A_2 = \pushup, A_3 = C, A_4 = N, A_5 = \pushup$.
Proposition~\ref{prop:getridofN} yields that
there exists a standardized formula $\phi'_0$ with 
$\phi \rarrstar_N \phi'_0$
where, with respect to the sequence $A'_1 = \pushup, A'_2 = C, A'_3 = \pushup$, 
the property given in part (1) of the proposition holds.
The formula $\phi'_0$ is a formula obtained from $\phi$ by iteratively taking each
instance of quantification $Qx$, and renaming the variable $x$
to a variable that appears nowhere else.  So, for example,
by starting from $\phi$ and renaming $w$ as $a$ and then $u$ as $b$,
one could obtain 
$\phi'_0 = 
 (\exists a (E(u,a) \land E(a,v))) \land (\exists b (E(v,b) \land E(b,w)))$.
The proposition then says that, by applying the sequence of rules
$A'_1, A'_2, A'_3$ at the locations where the respective rules 
$A_2, A_3, A_5$ were applied, we obtain a formula
$\phi'_3 = 
\exists a \exists b \hspace{3pt} ( 
(
(E(v,b) \land E(b,w))
\land 
(E(u,a) \land E(a,v))
)$
from which we can obtain the last formula $\phi_5$ in the given sequence
solely by using the renaming rule, which evidently holds in this case.
\end{example}

\subsection{Holey formulas}

We next define the notion of \emph{holey formula}, which will be used
subsequently to define the notion of \emph{organized formula}.
An organized formula is a formula that can be naturally stratified
into \emph{regions}, where each region is a holey formula.

A \emph{holey formula} is intuitively defined similarly to a formula,
but whereas a formula contains of atoms, a holey formula has placeholders where further formulas
can be attached; these placeholders are represented by natural numbers.
Formally,
a \emph{holey pfo-formula} is a formula built as follows:
each natural number $i \ge 1$ is a holey pfo-formula,
and is referred to as a \emph{hole};
when~$\phi$ and $\phi'$ are holey pfo-formulas,
so are $\phi \wedge \phi'$ and $\phi \vee \phi'$; and,
when~$\phi$ is a holey pfo-formula and $x$ is a variable,
so are $\exists x \phi$ and $\forall x \phi$.
We require that for each holey pfo-formula~$\phi$, no natural number
occurs more than once.
We say that a holey pfo-formula is \emph{atomic}
if it is equal to a natural number (equivalently, if it contains no
connectives nor quantifiers).
A
\emph{holey $\exwe$-formula}
is defined as a holey pfo-formula where, apart from atoms,
only the formation rules involving
existential quantification and conjunction are permitted.
A \emph{holey $\fove$-formula} is defined dually.

When $\phi$ is a holey formula with holes among $1, \ldots, k$ and
we have that $\psi_1, \ldots, \psi_k$ are formulas, we use
$\phi \llb \psi_1, \ldots, \psi_k \rrb$ to denote the formula obtained from $\phi$
by substituting, for each $i = 1, \ldots, k$, the formula $\psi_i$ in place
of $i$.
\begin{example}
  Consider the holey pfo-formula 
$\phi = 2 \wedge 1$, and the formulas $\psi_1 = S(x) \vee S(y)$,
  $\psi_2 = R(x,z)$.  We have
$\phi \llb \psi_1, \psi_2 \rrb = R(x,z) \wedge (S(x) \vee S(y))$.
\end{example}

In the sequel, we will speak of applying rewriting rules
(generally excluding the rule~$N$) to
holey formulas.  In order to speak of the applicability of rules such as the pushdown and
pushup rules, we need to associate a set of free variables to each subformula of a holey formula.
We define an \emph{association} for a holey formula $\phi$ to be a partial mapping $a$
defined on the natural numbers that is defined on each hole in $\phi$,
and where, for each number $i$ on which $a$ is defined, it holds that $a(i)$ is
a set of variables.
With an association $a$ for a holey formula $\phi$, we can naturally define a set
of free variables on each subformula of $\phi$, by considering
$\free(i) = a(i)$ for each~$i$ on which $a$ is defined, and then using the usual
inductive definition of $\free(\cdot)$.

\subsection{Organized formulas}

\newcommand{\regions}{\mathrm{regions}}

We next define and study \emph{organized formulas}, which, intuitively speaking,
are pfo-formulas that are stratified into regions, 
based on the quantifiers and connectives.  %

We define \emph{$\exwe$-organized} formulas and \emph{$\fove$-organized}
formulas
by mutual induction, as follows.
\begin{itemize}
\item When $\phi$ is a non-atomic holey $\exwe$-formula and each of
  $\psi_1, \ldots, \psi_k$ is an atom or a $\fove$-organized formula,
  then $\phi \llb \psi_1, \ldots, \psi_k \rrb$ is an $\exwe$-organized formula.
\item When $\phi$ is a non-atomic holey $\fove$-formula and each of
  $\psi_1, \ldots, \psi_k$ is an atom or an $\exwe$-organized formula,
  then $\phi \llb \psi_1, \ldots, \psi_k \rrb$ is a $\fove$-organized formula.
  
\end{itemize}
An \emph{organized formula} is defined to be a formula
that is either an $\exwe$-organized formula or a $\fove$-organized formula.
An organized formula can be naturally decomposed into \emph{regions}:
essentially, when we have an organized formula 
 $\theta = \phi \llb \psi_1, \ldots, \psi_k \rrb$, 
the holey formula $\phi$ is a region of $\theta$, and 
for whichever of the formulas $\psi_1, \ldots, \psi_k$
are organized and appear in $\phi$, their regions are also regions of $\theta$.
Formally, we recursively define the set of regions of an organized formula 
$\theta = \phi \llb \psi_1, \ldots, \psi_k \rrb$,
denoted $\regions(\theta)$, as follows:
$\regions(\theta) = 
\{ \phi \} \cup \bigcup_{\text{$i$ occurs in $\phi$}} \regions(\psi_i)$,
where the union is over all holes $i$ that appear in $\theta$.

\begin{example}
Consider the formula
$\theta = (\forall y R(x,y)) \lor (\exists z (S(x,z) \land T(z)))$.
We can view this formula as a $\fove$-organized formula, in the following way.
Define $\beta$ to be the $\fove$-holey formula $(\forall y 1) \lor 2$,
and define  $\gamma$ to be the $\exwe$-holey formula 
$\exists x (1 \land 2)$.  
Then, it is straightforwardly verified that
 $\theta = \beta\llb R(x,y), \gamma\llb S(x,z), T(z)  \rrb  \rrb$.
\end{example}

The next proposition implies that every pfo-formula (that is not an atom) is an organized formula.
Define the \emph{top operation} of a pfo-formula to be the label of the root of the formula when interpreting it as a tree.

\begin{propositionrep}
\label{prop:each-pfo-is-organized}
  Each non-atomic pfo-formula $\theta$ is an organized formula. More precisely,  each non-atomic pfo-formula $\theta$ whose top operation is $\exists$ or $\land$ is an $\exwe$-organized formula. Each non-atomic pfo-formula $\theta$ whose top operation is $\forall$ or $\lor$ is an $\fove$-organized formula.
\end{propositionrep}
\begin{proof}
    We define a \emph{switch} along a path in $\theta$ to be an edge $vu$ such that the label of $v$ is in $\exwe$ and the label of $u$ is in $\fove$ or the label of $v$ is in $\fove$ and that of $u$ is in $\exwe$. Then the switch-depth of a pfo-formula $\theta$ is defined as the maximal number of switches on a root-leaf path taken over all root-leaf paths in $\theta$. We show the claim by induction on the switch-depth of~$\theta$.
    
    If the switch-depth is $0$, then $\theta$ is either a $\exwe$-formula or a $\fove$-formula. By definition of organized formulas, the claim is true in that case.
    
    For the case of positive switch-depth, consider first the case that the top operation $\theta$ is in $\exwe$. 
We construct a holey $\exwe$-formula $\theta'$ as follows: let $T'$ be the largest subtree of~$\theta$ that contains the root of $\theta$ and whose nodes all have labels in $\exwe$. 
Let $v_1, \ldots, v_k$ be the nodes in $\theta$ that have their parent in $T'$ but are themselves not in $T'$. Let $\psi_1, \ldots, \psi_k$ be the subformulas such that for every $i\in [k]$ the node $v_i$ is the root of $\psi_i$. We construct $\theta'$ by substituting each $v_i$ by the hole~$i$. 
Clearly, $\theta'$ is a holey non-atomic $\exwe$-formula. Moreover, the $\psi_i$ are all either atoms or non-atomic pfo-formula whose top operation is in $\fove$. Also, their switch-depth is lower than that of $\theta$. 
So, by induction, we get that the $\psi_i$ are all atoms or $\fove$-organized formulas and thus, by definition, $\theta= \theta'\llb \psi_1, \ldots, \psi_k \rrb$ is a $\exwe$-formula, as claimed.
    
    If $\theta$ has its top operation in $\fove$, the claim follows by symmetric reasoning.
\end{proof}

When the rules in $\T \setminus \{ N \}$ are applied to
organized formulas, they act on regions independently, in the following
formal sense.

\begin{propositionrep}
  \label{prop:T-to-organized}
  Suppose that $\Phi = \phi \llb \psi_1, \ldots, \psi_k \rrb$ is an organized formula.
  Then, each formula in $[\Phi]_{\T \setminus \{ N \}}$
  has the form $\phi' \llb \psi'_1, \ldots, \psi'_k \rrb$,
  where $\phi' \in [\phi]_{\T \setminus \{ N \}}$ and
  $\psi'_1 \in [\psi_1]_{\T \setminus \{ N \}}, \ldots,$ $\psi'_k \in [\psi_k]_{\T \setminus \{ N \}}$.
  Here, we understand the $\T \setminus \{ N \}$-rules to be applied to $\phi$ and holey pfo-formulas
  under the association $i \mapsto \free(\psi_i)$ defined on each $i = 1, \ldots, k$.
\end{propositionrep}
\begin{proof}
    We prove the slightly stronger statement that if~$\phi$ is a holey $\exwe$-formula then so is $\phi'$ and all $\psi_i'$ are atomic or $\fove$-organized formulas. Symmetrically, if~$\phi$ is a holey $\fove$-formula then so is $\phi'$ and all $\psi_i'$ are atomic or $\exwe$-organized formulas. 
    We perform structural induction on the formula $\phi\llb \psi_1, \ldots, \psi_k \rrb$. So assume first that all $\psi_i$ in $\phi\llb \psi_1, \ldots, \psi_k \rrb$ are atoms. Then the statement is clear since for all $i\in [k]$ we have that $[\psi_i]_{\T\setminus \{N\}}$ consists only of a single atomic formula and thus all rule applications are on $\phi$.
    
    Now let the $\psi_i$ be general, not necessarily atomic formulas. Assume that $\phi$ is a $\fove$-formula; the case of $\exwe$-formulas is analogous. Then we know that all $\psi_i$ are atoms or $\exwe$-organized formulas and, by induction, all $\psi_i'\in [\psi_i]_{\T\setminus \{N\}}$ are for non-atomic $\psi_i$ also $\exwe$-organized formulas. Thus, whenever applying any sequence of rules in $\rarr_{\T\setminus \{N\}}$ onto $\psi_i$, the top operation of the resulting formula~$\psi_i'$ will be $\exists$ or $\land$. Moreover, whenever applying rules from $\rarr_{\T\setminus \{N\}}$ in such a way that only nodes in $\phi$ are concerned, the top operation of $\phi'$ is still in $\fove$. Finally, we have that the node above any hole~$i$ in $\phi$ is either $\forall$ or $\lor$ while the top operation of $\psi_i$ is in $\exwe$ which by induction remains true after applying rules from $\rarr_{\T\setminus \{N\}}$. However, no rule in $\T\setminus \{N\}$ can involve both a node with label in $\fove$ and one with label in $\exwe$, and thus all applications involve either only the holey formula $\phi$ or one of the $\psi_i$ but never nodes of both. So the proposition follows directly.
\end{proof}

\section{Rule applicability}
\label{sect:rule-applicability}

In the remainder of the paper, it will be crucial to have an understanding for when the rules of the systems $Y$ and $Y'$ can be applied. To this end, in this section, we study the structure of formulas that allow their application. 
This will in particular also lead to the normal-form result of Theorem~\ref{thm:normal-forms}.

Throughout this section, we make use of multisets; for each multiset $S$, we use the notation $|S|$ to denote the size of $S$ as a multiset, that is, the size of $S$ where repetitions are counted.

We start with the applicability of the removal rule $\rarr_{M}$. Remember that we say that a rule $\rarr$ can be applied to a set $\Phi$ of formulas if and only if there is a formula $\phi\in \Phi$ 
to which the rule can be applied.

\begin{lemmarep}\label{lem:applyM}
	Let $\theta$ be a pfo-formula. Then the following statements are equivalent:
	\begin{itemize}
		\item $\rarr_M$ can be applied to $\theta$.
		\item $\rarr_M$ can be applied to $[\theta]_{ACO}$.
		\item $\rarr_M$ can be applied to $[\theta]_{\T}$.
	\end{itemize}
\end{lemmarep}
\begin{proof}
	The proof is based on the simple observation that a quantifier does not bind a variable in $\theta$ if and only if it does not bind a variable in any representative in $[\theta]_{ACO}$ and $[\theta]_{\T}$.

    If $\rarr_M$ can be applied to $\theta$, then it can by definition also be applied to $[\theta]_{ACO}$ since $\theta\in [\theta]_{ACO}$. Moreover, since $[\theta]_{ACO}\subseteq [\theta]_{\T}$, the rule $\rarr_M$ is applicable to $[\theta]_{\T}$ if it is applicable to $[\theta]_{ACO}$.
    
    It remains to show that if $\rarr_M$ is applicable to $[\theta]_{\T}$, then it is also applicable to $\theta$.
    To see this, observe first that it is straightforward to verify that the rules in $\T$ preserve the number of instances
    of quantification, and also, for any instance $Q x$ of quantification, the rules in $\T$
    preserve the number of occurrences of the variable $x$ that are bound to $Q x$.
    Hence, if $\rarr_{M}$ can be applied to an instance $Q x$ of quantification
    in a formula in $[\theta]_{\T}$, there is a corresponding instance $Q x$ of quantification
    in $\theta$ to which $\rarr_{M}$ can be applied, and thus $\rarr_{M}$
    can be applied to $\theta$.
\end{proof}

It will be useful to consider (sets of) formulas in which the pushdown operation has been applied exhaustively, in particular, to understand 
normal forms of the system $Y$.
To this end, we introduce the following definitions.
We say that a pfo-formula $\phi$ is \emph{pushed-down}
if $\rarr_{\pushdown}$ cannot be applied to~$\phi$;
we say that a set $\Phi$ of pfo-formulas is \emph{pushed-down}
if each formula $\phi \in \Phi$ is pushed-down.

\newcommand{\conjs}{\mathrm{conjuncts}}

We say that a pfo-formula or a holey pfo-formula is a \emph{$k$-fold conjunction}
if, up to associativity, it can be written in the form
$\psi_1 \wedge \cdots \wedge \psi_k$.  We can formalize this as follows.
For each pfo-formula $\theta$, define the multiset $\conjs(\theta)$
inductively, as follows.  If~$\theta$ is an atom, a hole, or begins
with disjunction or quantification, define
$\conjs(\theta) = \{ \theta \}$, where $\theta$ has multiplicity $1$.
When $\theta$ has the form $\theta_1 \wedge \theta_2$, define
$\conjs(\theta)$ as the multiset union $\conjs(\theta_1) \cup \conjs(\theta_2)$.
We say that $\theta$ is a \emph{$k$-fold conjunction} when
$|\conjs(\theta)| \ge k$. We define \emph{$k$-fold disjunctions} analogously.

The notion of $k$-fold conjunction will be useful in the remainder as it allows us to understand when we can apply the splitdown rule on quantifiers. The following lemma tells us that the existence of $k$-fold conjunctions in our equivalence classes with respect to $\{A,C,O\}$ and $\T$ can be decided by looking at any pushed-down representative.

\begin{lemmarep}
	\label{lemma:n-fold}
	Suppose that $\theta$ is a pfo-formula with
	$[\theta]_{ACO}$  pushed-down, and let $k \geq 1$.
	Then, there exists a $k$-fold conjunction in $[\theta]_{\T}$
	if and only if $\theta$ is a $k$-fold conjunction.
	Similarly, there exists a $k$-fold disjunction
	in $[\theta]_{\T}$
	if and only if $\theta$ is a $k$-fold disjunction.
\end{lemmarep}
\newcommand{\msett}{\mathrm{mset}_{\T}}
\begin{proof}
    The second statement is dual to the first; we prove the first.
    The backward implication is immediate, so we prove the forward implication.
    It is directly true when $\theta$ is an atom, so we assume that $\theta$ is non-atomic.
    As renaming preserves the property of being an $k$-fold conjunction,
    by appeal to Proposition~\ref{prop:getridofN}, we may assume that $\theta$ is standardized,
    and that there exists a $k$-fold conjunction in $[\theta]_{\T \setminus \{ N \}}$.
    By appeal to
    Proposition~\ref{prop:each-pfo-is-organized},
    we view $\theta$ as an organized formula
    $\phi \llb \psi_1, \ldots, \psi_k \rrb$.
    
    In the case that $\phi$ is a holey $\{ \forall, \vee \}$-formula, then by Proposition~\ref{prop:T-to-organized}, 
    any rewriting of $\theta$ via the rules in $\T \setminus \{N\}$ yields a formula that (up to renaming)
    is of the form
    $\phi' \llb \psi'_1, \ldots, \psi'_k \rrb$ where $\phi'$ is a $\T \setminus \{N\}$-rewriting of $\phi$,
    and hence
    also a holey $\{\forall,\vee\}$-formula;
    hence, every formula in $[\theta]_{\T \setminus \{N\}}$ is a $1$-fold conjunction.
    We thus suppose that $\phi$ is a holey $\{ \exists, \wedge \}$-formula.

    For every holey $\{ \exists, \wedge \}$-formula $\alpha$ with holes among
    $1, \ldots, k$, we define
    $\msett(\alpha)$ as a multiset
    whose elements are sets of variables.  This multiset is defined inductively, as follows:
    \begin{itemize}
        \item $\msett(i) = \{ \free(\psi_i) \}$, when $i \in \nats$,
        \item $\msett(\phi_1 \wedge \phi_2) = \msett(\phi_1) \cup \msett(\phi_2)$,
        \item $\msett( \exists x \phi_1 ) = \Join_{(x)} \msett(\phi_1)$.
    \end{itemize}
    In the first item, the single element is understood to have multiplicity $1$; in the second item, the union is understood as a multiset union; in the third item, we understand $\Join_{(x)} S$ to be the multiset obtained from $S$ by removing each element that contains $x$, that is, by removing from $S$ each
    element that is contained in  $\mathcal{U} = \{ U \in S ~|~ x \in U \}$,
    and then adding the set $(\bigcup_{U \in \mathcal{U}} U) \setminus \{ x \}$, where here,
    the $\bigcup$ denotes the usual set union. Note that in this case 
    \begin{align*}|\msett(\exists x \phi_1)| = 1 + |\{U\in \msett(\phi_1)\mid x\notin U\}|.\end{align*}
    We use the straightforwardly verified fact that,
    for each set $U \in \msett(\alpha)$, it holds that $U \subseteq \free(\alpha)$.
    
    \begin{claim} \label{claim:invariance-msett}
    Supposing $\alpha$ to be a holey $\{ \exists, \wedge \}$-formula,
    $\msett(\alpha)$ is invariant under applications of rules from $\T \setminus \{ N \}$
    to $\alpha$. 
\end{claim}
\begin{proof} 
For the rules A and C, this follows from associativity and commutativity
    of multiset union, respectively; for the rules $\pushdown$ and $\pushup$,
    this follows from the fact that when $x$ does not appear in any element of $\msett(\phi_2)$,
    $(\Join_{(x)} \msett(\phi_1)) \cup \msett(\phi_2) =  \Join_{(x)} (\msett(\phi_1) \cup \msett(\phi_2))$.
    For the rule~$O$, this follows from the fact that for any two variables $x, y$,
    the operators $\Join_{(x)}$, $\Join_{(y)}$ commute.
    \end{proof}
    \begin{claim} \label{claim:writing-as-conjunction}
    For each subformula $\phi^-$ of $\phi$, setting $s$
    to be the size of the multiset $\msett(\phi^-)$, it holds that $\phi^-$ can be written,
    up to associativity,
    as
    a conjunction $\beta_1 \wedge \cdots \wedge \beta_s$ having $s$ conjuncts
    such that $\msett(\phi^-)$ is equal to the
    multiset $\{ \free(\beta_1), \ldots, \free(\beta_s) \}$.
\end{claim}
\begin{proof}
    Let us prove this by structural induction.  
It is immediate when $\phi^-$ is a hole
    or a conjunction $\phi_1 \wedge \phi_2$.
    When~$\phi^-$ has the form $\exists x \phi_1$,
    let $\phi_1 = \beta_1 \wedge \cdots \wedge \beta_s$ such that
    $\msett(\phi_1) = \{ \free(\beta_1), \ldots, \free(\beta_s) \}$.
    We observe that each element
    of $\msett(\phi_1)$ contains~$x$: if not, there exists an index~$j$ 
    such that $x \notin \free(\beta_j)$; by applying the rules A and C,
    $\phi_1$ can be expressed as a conjunction $\beta' \wedge \beta_j$ where
    $\beta'$ is the conjunction of all $\beta_i$ apart from $\beta_j$,
    and then the $\pushdown$ rule can be applied to $\beta' \wedge \beta_j$,
contradicting our assumption that $[\theta]_{ACO}$ is pushed-down.
    It follows that $\msett(\phi^-)$ has size~$1$,
    and that $\phi^-$
    is itself a $1$-fold conjunction having the desired form.
\end{proof}    

    By assumption,
    there exists a $k$-fold conjunction~$\theta^c$ in $[\theta]_{\T \setminus \{ N \}}$; by Proposition~\ref{prop:T-to-organized}
we can view $\theta^c$ as having the form $\phi^c\llb \psi^c_1, \ldots, \psi^c_k  \rrb$ where $\phi^c$ and $\phi$ can be rewritten into each other via
the rules in $\T \setminus \{ N \}$.
We have $k \leq |\msett(\phi^c)|$ by the definition of $\msett(\cdot)$.
By the invariance of $\msett(\cdot)$ established 
in Claim~\ref{claim:invariance-msett},
    we have $|\msett(\phi_c)| = |\msett(\phi)|$.
Hence $k \leq |\msett(\phi)|$.
By Claim~\ref{claim:writing-as-conjunction} applied in the case where $\phi^-$ is equal to $\phi$ itself, we have that $\phi$ is
a $k$-fold conjunction.
It follows that $\theta$ is a $k$-fold conjunction.
\end{proof}

As a consequence of Lemma~\ref{lemma:n-fold}, we get that for pushed-down formulas there is a result similar to Lemma~\ref{lem:applyM} that characterizes when the splitdown rule can be applied.

\begin{lemmarep}\label{lem:applyS}
	Let $\theta$ be a pfo-formula where $[\theta]_{ACO}$ is pushed-down. Then the following are equivalent:
	\begin{itemize}
		\item 	$\rarr_{\splitdown}$ can be applied to $\theta$.
		\item 	$\rarr_{\splitdown}$ can be applied to $[\theta]_{ACO}$.
		\item 	$\rarr_{\splitdown}$ can be applied to $[\theta]_{\T}$.
	\end{itemize}
\end{lemmarep}
\begin{proof}
	The implications from the first point to the second and from the second to the third are again directly clear as in the proof of Lemma~\ref{lem:applyM}.
	
	So assume now that $\rarr_{\splitdown}$ can be applied to $[\theta]_{\T}$.
	Then, by definition, there exists a formula $\theta^+ \in [\theta]_{\T}$
	to which $\rarr_{\splitdown}$ can be applied.
	Since applicability of $\rarr_{\splitdown}$ to a formula
	and whether or not $[\theta]_{ACO}$ is pushed-down
	are preserved by variable renaming, we may assume by Proposition~\ref{prop:getridofN}
	that $\theta$ is standardized, and that $\theta^+ \in [\theta]_{\T \setminus \{ N \}}$.  (This is justified by applying Proposition~\ref{prop:getridofN} to the formulas $\phi_0 = \theta$ and $\phi_t = \theta^+$, and then replacing $\theta$ with $\phi'_0$ and $\theta^+$ with $\phi'_s$.)
	Clearly, the formulas $\theta$ and $\theta^+$ are non-atomic.
	Observe that when $\rarr_{\splitdown}$ is applicable to an organized formula,
	the quantification must come from~a holey $\{ \exists, \wedge \}$-formula
	and the connective must come from a holey $\{ \forall, \vee \}$-formula,
	or vice-versa.
	
	By appeal to the decomposition of each non-atomic pfo-formula into an
	organized formula (Proposition~\ref{prop:each-pfo-is-organized}),
	and Proposition~\ref{prop:T-to-organized},
	there exist organized formulas
	$\phi \llb \psi_1, \ldots, \psi_k \rrb$,
	$\phi^+ \llb \psi^+_1, \ldots, \psi^+_k \rrb$ that
	are subformulas of $\theta$ and $\theta^+$, respectively,
	where 
\begin{enumerate}
\item $\phi \lrarrstar_{\T \setminus \{ N \}} \phi^+$ and
$\psi_1 \lrarrstar_{\T \setminus \{ N \}} \psi^+_1, \ldots, \psi_k \lrarrstar_{\T \setminus \{ N \}} \psi^+_k$;
	and  
\item  $\rarr_{\splitdown}$ can be applied to an instance of quantification
	in $\phi^+$ and a connective in a formula~$\psi^+_\ell$.
\end{enumerate}
	
	We assume up to duality that $\phi^+$ is a holey $\{ \forall, \vee \}$-formula
	and that $\psi^+_\ell$ is a $\{ \exists, \wedge \}$-organized formula.
	Viewing $\phi^+$ and $\psi^+_\ell$ as trees, we have that $\ell$ occurs as a leaf in $\phi^+$, that
	the parent of this leaf in $\phi^+$ is a universal quantification $\forall x$,
	and that $\psi^+_\ell$ has conjunction $(\wedge)$ at its root.
	Thus $\psi^+_\ell$ is a $2$-fold conjunction.
        Since $[\theta]_{ACO}$ is pushed-down, and $\psi_\ell$ is a subformula of $\theta$, we have that $[\psi_\ell]_{ACO}$ is pushed-down, 
so we can apply Lemma~\ref{lemma:n-fold}
	to obtain that $\psi_\ell$ is a $2$-fold conjunction.
	Let $B$ be the set of holes occurring in~$\phi$;
 since
	$\phi \lrarrstar_{\T \setminus \{ N \}} \phi^+$, we have that $B$ is also the set of
	holes occurring in $\phi^+$.
	We have $\free(\psi_i) = \free(\psi^+_i)$ for each $i = 1, \ldots, k$.
	We know that we can apply rules in $\T \setminus \{ N \}$ to obtain $\phi$ from $\phi^+$; consider the process where we obtain $\phi$ from $\phi^+$ by applying such rules one at a time.  
For each $i \in B$, if $x$ occurs free in $\psi_i$, then whether or not $x$ is bound to the mentioned quantification $\forall x$ is preserved throughout the process.
Thus, among the formulas $\psi_i$ ($i \in B$), only $\psi_\ell$
	can have a free occurrence bound to the corresponding quantification $\forall x$ in $\phi$.
	It thus follows that, in $\phi$ viewed as a tree, there is no instance of disjunction
	between $\forall x$ and $\ell$, for if there were, by quantifier reordering (applying $\rarr_O$),
	$\forall x$ could be moved above a disjunction, and then $\rarr_{\pushdown}$ would be applicable
	(contradicting that $[\theta]_{ACO}$ is pushed-down).
	Thus, in $\phi$ viewed as a tree, the parent of $\ell$ is a universal quantification,
	and so $\rarr_{\splitdown}$ can be applied to $\theta$.
\end{proof}

With the insights on rule applicability from above, we now show that there is a polynomial time algorithm that decides if the defining rules of the system $Y=(\pfo, \to_{\{\pushdown,\splitdown,M\}}) \div \lrarrstar_{ACO}$ can be applied to an equivalence class with respect to $ACO$. 
This will be a building block in the proof of Theorem~\ref{thm:push-in} 
(proved in the next section)
that shows that we can compute normal forms in the system~$Y$ in polynomial time.

\begin{lemma}\label{lem:testruleapplication}
	There is a polynomial time algorithm that, given a pfo-formula $\theta$, decides if there is a formula $\theta'\in [\theta]_{ACO}$ such that
one of the rules $\rarr_{M}$, $\rarr_{\pushdown}$, $\rarr_{\splitdown}$,
can be applied to it. If so, it computes
such a $\theta'$ and a formula $\theta''$ 
obtainable by applying one of these rules to~$\theta'$.
\end{lemma}
\begin{proof}
	We consider the three rules in $Y$ in sequence, starting with $\rarr_{M}$, then $\rarr_{\pushdown}$ and finally $\rarr_{\splitdown}$.
	
	For $\rarr_{M}$, by Lemma~\ref{lem:applyM}, a quantifier can be deleted in $\theta$ if and only if one can be deleted from any $\theta'\in [\theta]_{ACO}$; we thus directly get a polynomial-time algorithm for $\rarr_{M}$.
	
	Now assume that $\rarr_{M}$ cannot be applied, so every quantifier binds a variable in at least one atom in its subformula in~$\theta$. 
Let us study when the rule $\rarr_{\pushdown}$ can be applied.
Fix one node~$v$ in the syntax tree of $\theta$, corresponding to a quantification, and say that the corresponding quantification is universal, say, $\forall x$; the other case is totally analogous. 
Let $\theta_v$ denote the subformula of $\theta$ rooted at $v$.
We assume that the quantifiers are checked inductively in a bottom-up fashion, so all quantifiers in strict subformulas of $\theta_v$ have been checked before dealing with $v$. So we assume we cannot apply $\rarr_{\pushdown}$ on any of them for any $\theta'\in [\theta_v]_{ACO}$. By Proposition~\ref{prop:each-pfo-is-organized}, we have that the subformula $\theta_v$ is organized. Let $\psi_v$ be the top region of $\theta_v$, so that $\psi_v$ is a holey $\fove$-subformula.
Now, in $\psi_v$, walk down the syntax tree starting from $v$ until 
we encounter a node $w$ that does not correspond to a universal quantification,
and let $\psi$ denote the holey subformula of $\psi_v$ rooted at that node $w$.
If $\psi$ is itself a hole, then clearly we cannot apply $\rarr_{\pushdown}$ on $v$ in $\theta$ as then $\psi$ does not contain any $\lor$-operation on which we could apply the pushdown; indeed this also holds for all formulas $\theta'\in [\theta]_{ACO}$ since the corresponding holey formula $\psi'$ in $\theta'$ also contains no $\lor$-operation.  
In the case that $\psi$ is not a hole, then we must have that $\psi$
begins with a disjunction; note that in this case,
we could apply the reordering rule to place the quantification of $v$
immediately on top of $\psi$.
Collapsing $\psi$ up to associativity,
we can write $\psi$ as an $r$-fold disjunction
	\begin{align}
		\psi= \bigvee_{i=1}^r \psi_i,\label{eq:decomposepsi}
	\end{align}
	for some integer $r$, where each $\psi_i$ 
is either a hole or has a universal quantifier on top. 
Now if there is $i^*\in [r]$ such that $\psi_{i^*}$ does not contain the variable $x$ (quantified at the node~$v$) in any of its holes, we can use the rules $A,C$ to rewrite 
$\psi = \psi_{i^*} \lor \bigvee_{i\in [r]\setminus \{i^*\}} \psi_i$
	which gives us an $\lor$-operation to which we can apply $\rarr_{\pushdown}$ for the node $v$ (up to applications of the O rule). On the other hand, if all $\psi_i$ contain the variable $x$ in one of their holes, then this is the case for all $\psi'$ that we derive using the same process
for each formula $\theta' \in [\theta]_{ACO}$, 
because the ACO rules do not allow exchanging the positions of $\lor$ and any other operators. Thus, the representation (\ref{eq:decomposepsi}) for all $\psi'$ is the same as for $\psi$ up to the fact that the ACO rules might have been used on the disjuncts $\psi_i$. But in that case, in no $\psi'$ and thus in no $\theta'$ there is a $\lor$-operation that $v$ could be pushed over. This directly yields a polynomial time algorithm for this case: check nodes $v$ corresponding to quantifications from the bottom up (so that any node is checked before any proper ancestor of the node), and in each case find the corresponding formula $\psi$; when this formula $\psi$
is not a hole, view it in the form (\ref{eq:decomposepsi}) above, and then check to see if there is a hole that does not contain the variable $x$; if so, then the pushdown rule can be applied, up to the ACO rules, and if not,
proceed to the next variable.
	
	Finally, assume that $\rarr_{M}$ and $\rarr_{\pushdown}$ can both not be applied. Then in particular $\theta$ is pushed-down. Then, by Lemma~\ref{lem:applyS}, we get that $\rarr_{\splitdown}$ can be applied to $[\theta]_{ACO}$ if and only if it can be applied to $\theta$ which directly yields a polynomial time algorithm.
	
	Overall, we have polynomial time algorithms for all three rules under study. 
Observing that if any rule can be applied, we can compute a formula $\theta'$ and the result $\theta''$ of the rule application (as specified in the lemma statement) efficiently; this completes the proof.
\end{proof}

Finally, we give the proof of Theorem~\ref{thm:normal-forms} which we restate for the convenience of the reader. 
Recall that 
$$	Y = (\pfo, \to_{\{\pushdown,\splitdown,M\}}) \div \lrarrstar_{ACO},\quad
	Y' = (\pfo, \to_{\{\splitdown,M\}}) \div \lrarrstar_{\T}.$$
\theoremnormalforms*
\begin{proof}
Suppose that $[\phi]_{ACO}$ is a normal form of the system $Y$.
Then, we have that $[\phi]_{ACO}$ is pushed-down, since $\pushdown$
is a rule in $Y$.
Since  $[\phi]_{ACO}$ is a normal form of $Y$, 
neither of the rules $\splitdown$, $M$ can be applied to
 $[\phi]_{ACO}$, implying by 
Lemmas~\ref{lem:applyM} and~\ref{lem:applyS}
that neither of these two rules can be applied to  $[\phi]_{\T}$.
Thus $[\phi]_{\T}$ is a normal form of~$Y'$. 
\end{proof}

\section{Termination and confluence}
\label{sect:termination-confluence}

In this section, we will show that both systems $Y$ and $Y'$ are terminating. With Lemma~\ref{lem:testruleapplication} from the last section, this will yield Theorem~\ref{thm:push-in}, showing that we can efficiently compute normal forms for the system $Y$. For $Y'$, we will also show that it is locally confluent which then implies Theorem~\ref{thm:convergent}, the convergence of $Y'$.

We will start with termination for the system $Y$. Let us for every pfo-formula $\phi$ denote by $|\phi|$ the number of nodes in the syntax tree of $\phi$.

\begin{lemmarep}\label{lem:terminatingY}
  For every pfo-formula $\phi$, any reduction chain in the system $Y$ starting in $[\phi]_{ACO}$ terminates and has length at most  $|\phi|^3$.
\end{lemmarep}

\begin{proof}
    Consider a pfo-formula $\phi$. We will show that any chain of $\rarr_{\pushdown,\splitdown,M}$-steps starting in $[\phi]_{ACO}$ is upper bounded by $|\phi|^3$. To this end, we define a potential function $p$ on all pfo-formulas as follows: let $F$ be a subformula of $\phi$ whose root is labeled by a quantifier. Then the local potential $\bar p(F)$ of $F$ is defined as $a^2$ where $a$ is the number of atoms in $F$. Note that $\bar p(F)$ is positive. Then the potential of $\phi$ is the sum of the local potentials of all subformulas that have a quantifier labeling their root.
    
    We claim that applying any operation in $\rarr_{\pushdown,\splitdown,M}$ decreases the potential. So let $\phi'$ be obtained from $\phi$ by applying one reduction step. We consider the different possible cases:
    \begin{itemize}
        \item $\mathbf{\rarr_{\pushdown}}$: Let $F= \exists x (F_1 \land F_2)$ be a sub-formula of $\phi$ such that in $F_2$ the variable $x$ is not free, and consider the application 
$\exists x (F_1 \land F_2)\rarr_{\pushdown} \underbrace{(\exists xF_1)\land F_2}_{:=F'}$.
        The only change in the potential of $\phi$ when applying this rule is the change from $\bar p(F)$ to $\bar p(F')$, since no other sub-formulas of $\phi$ change. But $F_1$ contains fewer atoms than $F_1\land F_2$, so $\bar p(F) > \bar p(F')$, so the potential of $\phi$ decreases. All other cases for $\rarr{\pushdown}$ follow analogously.
        \item $\mathbf{\rarr_{\splitdown}}$:  Let $F= \forall x (F_1 \land F_2)$ be a sub-formula of $\phi$ and consider the application
$\forall x (F_1 \land F_2)\rarr_{\splitdown} (\forall x F_1)\land (\forall x F_2)$.
        Then this application changes the potential $p$ of $\phi$ by $\bar p (\forall x F_1) +\bar p (\forall x F_2) - \bar p (F)$. Let $a_1$ be the number of atoms in $F_1$ and $a_2$ the number of atoms in $F_2$, then
        \begin{align*}
            \bar p (\forall x F_1) +\bar p (\forall x F_2) = a_1^2 + a_2^2 < (a_1+a_2)^2 = \bar p(F),
        \end{align*}
        so the potential of $F$ decreases. All other cases for $\rarr_{\splitdown}$ follow analogously.
        \item $\mathbf{\rarr_{M}}$:  In the definition of $p(F)$, we lose one positive summand whenever applying $\rarr_{M}$ without changing any of the other summands, so the potential decreases.
    \end{itemize}
    So in any case, whenever applying one of the rules on $\phi$, the potential $p$ decreases.
    
    Now consider the equivalence classes of $PFO\div \lrarrstar_{ACO}$. Note that applying $\rarr_{A}$, $\rarr_C$ or $\rarr_O$ does not change the potential of any formula, so we can define the potential $[p]$ of every equivalence class $[\phi]_{ACO}$ by $[p]\left([\phi]_{{ACO}}\right) := p(\phi)$. Now, whenever applying a rule of the system to an equivalence class $[\phi]_{{ACO}}$, the potential $[p]([\phi]_{{ACO}})$ decreases, because, as we have seen before, it decreases for any $\phi'\in [\phi]_{{ACO}}$. 
    
    The potential is bounded by $[p]([\phi]_{{ACO}})\le |\phi| a^2\le |\phi|^3$ where $|\phi|$ is the length of~$\phi$ and $a$ the number of atoms. Since the potential is a positive integer, the longest descending chain starting in $[\phi]_{{ACO}}$ is thus of length $|\phi|^3$, so any chain in the system has polynomially bounded length.
 \end{proof}

The proof of Theorem~\ref{thm:push-in} follows from this Lemma~\ref{lem:testruleapplication} and Lemma~\ref{lem:terminatingY} straightforwardly.
\begin{proof}[Proof of Theorem~\ref{thm:push-in}]
    Starting in an input $\theta$, we simply run the algorithm of Lemma~\ref{lem:testruleapplication} to apply the rules of $Y$ as long as we can apply one of them. Let $\theta_i$ be the representative that the algorithm gives us after $i$ applications. Then, by the definition of the rules, the size when going from $\theta_i$ to $\theta_{i+1}$ increases by one (in the case of $\rarr_{\splitdown}$) or stays the same. Noting that we can only apply rules a polynomial number of times by Lemma~\ref{lem:terminatingY}, we get that all formulas we have to consider have size polynomial in $|\theta|$. Thus, after polynomial time, the algorithm of Lemma~\ref{lem:testruleapplication} will report that none of the rules can be applied anymore in which case we have found a normal form for $[\theta]_{ACO}$ in the system $Y$.
\end{proof}

We next turn to the termination of $Y'$, which will be used to show Theorem~\ref{thm:convergent} with Lemma~\ref{lemma:newman}.

\begin{lemma}\label{lem:terminatingYp}
	The system $Y'$ is terminating.
\end{lemma}
\begin{proof}
    The proof follows the same idea as that of Lemma~\ref{lem:terminatingY}: we want to introduce a potential for every pfo-formula and show that it is the same for all representatives of an equivalence class. Then we would like to show that applying rules from the system $Y'$ decreases the potential of any formula. Unfortunately, the setting is slightly more complicated such that the definition of the potential is more involved.
    
    Consider a pfo-formula $\theta$. Let $v$ be a node labeled by a universal quantification and let~$F$ be the subformula of $\theta$ rooted in $v$. We say that an ancestor $u$ of $v$ is \emph{blocking} for $v$ if it is an $\land$-node or if it is of the form $\exists y$ and this quantifier binds an occurrence of $y$ below $u$. We call all other nodes \emph{non-blocking}. The \emph{non-blocked ancestor} $v^*$ of $v$ is the highest node in $\theta$ reachable from $v$ by using only non-blocking nodes. Observe that $v^*$ is uniquely determined because $\theta$ is a tree. The local potential $\bar p(F)$ is then the number of nodes labeled by $\land$ that appear in the subformula of $\theta$ rooted by $v^*$. To simplify notation, we sometimes also write $\bar p(v)$ for $\bar p(F)$. We define the local potential of nodes labeled with an existential quantifier analogously, letting $\lor$ and $\forall$ play the role of $\land$ and $\exists$ in the definition of non-blocking nodes. Thus we have defined the local potential for all nodes labeled by quantifications; all other nodes have no local potential.
    
    Then the potential $p(\theta) := \sum_F \left(3 \bar p(F) + \mu(F)\right)$ where the sum is over all subformulas $F$ of $\theta$ whose root is labeled by~a quantifier and 
    \begin{align*}
        \mu(F) = \begin{cases}
            1, &\text{ if } \bar{p}(F) = 0\\
            0, &\text{otherwise}.
        \end{cases}
    \end{align*}
    
    \begin{claim}\label{clm:same-potential}
        All formulas $\theta'\in [\theta]_\T$ have the same potential $p(\theta')=p(\theta)$.
    \end{claim}
    \begin{proof}
        It suffices to show that applying any rule in $\T$ does not change the potential. Note that applying any rule in~$\T$ does not delete or add any quantifiers, so it suffices to show that $\bar p(v)$ does not change for any node $v$. We only consider the case in which~$v$ is labeled by a universal quantifier; the other case is symmetric.
        So for the rest of the proof of the claim fix a node $v$ that is labeled by a universal quantification. We consider applications of the different rules in~$\T$.
        
        \begin{itemize}[align=left] 
            \item[$\mathbf{\rarr_{C}}$:] Applying commutativity does not change the non-blocked ancestor of any node $v$, nor the $\land$-nodes below it. So clearly $\bar p(v)$ does not change.
            \item[$\mathbf{\rarr_{A}}$:] Applying associativity might change the non-blocked ancestor $v^*$. However, the $\land$-nodes below the respective non-blocked ancestors stay the same, so $\bar p(v)$ stays unchanged.
            \item[$\mathbf{\rarr_{O}}$:] Since $\forall$-quantifiers are non-blocking, reordering them---while it may change the non-blocked ancestor---will not change the set of $\land$-nodes below that ancestor. $\exists$-quantifiers might be blocking, so reordering them might or might not change the non-blocked ancestor. But in any case, the $\land$-nodes under the non-blocked ancestor stay the same and thus $\bar p(v)$ is unchanged.
            \item[$\mathbf{\rarr_{N}}$:] The definition of $\bar p(v)$ is completely independent of variable names, so renaming does not change the potential $\bar p(v)$.
            \item[$\mathbf{\rarr_{\pushdown}}$:] Consider first the case that $\rarr_{\pushdown}$ was applied on a $\forall$-quantifier, so we have an operation of the form 
            \begin{align*}
                \forall x (F_1 \vee F_2) \rarr_{\pushdown} (\forall x F_1) \vee F_2.
            \end{align*}
            Call the node labeled with $\forall$ before this application $u$ and let the $\vee$-node below $u$ be $u'$.
            Both $\forall$ and $\lor$ are non-blocking, so if $u$ is not the non-blocked ancestor of $v$, then $v^*$ does not change and thus $\bar p(v)$ remains unchanged as well. If $u$ is the unblocked ancestor of $v$, then pushing $u$ down below $u'$ makes $u'$ the new non-blocked ancestor. However, before and after the application of the rule all $\land$-nodes below the non-blocked ancestor are actually below $u'$, so $\bar p(v)$ remains unchanged by the application.
            
            Now assume that $\rarr_{\pushdown}$ is applied on an $\exists$-quantifier, so we have an operation of the form 
            \begin{align*}
                \exists x (F_1 \land F_2) \rarr_{\pushdown} (\exists x F_1) \land F_2.
            \end{align*}
            Let $u$ be the node labeled by $\exists x$ before the application of the rule, $u'$ its $\land$-child and $u''$ the root of $F_1$. Note that $\land$ is blocking, so the path from $v$ to $v^*$ cannot pass through $u'$. However, the node $u$ might be non-blocking for $v$, so it can in particular become the non-blocked ancestor of $v$ after the application of the rule. This case happens only if before application the node $u''$ was the non-blocked ancestor of $v$, in which case the $\land$-nodes below the non-blocked ancestor stay the same. If $u$ does not become the non-blocked ancestor, then that ancestor remains the same and the $\land$-nodes below it also do not change. Thus, in any case $\bar p(v)$ does not change.
            \item[$\mathbf{\rarr_{\pushup}}$:] Since $\mathbf{\rarr_{\pushup}}$ is the inverse of $\mathbf{\rarr_{\pushdown}}$ and the latter does not change $\bar p(v)$, this is also true for $\mathbf{\rarr_{\pushup}}$.
        \end{itemize}
        Since we have checked all rules in $\T$ and none of them change the potential, this completes the proof of the claim.
    \end{proof}
    
    We continue the proof of Lemma~\ref{lem:terminatingYp}. Since the potential is the same for all pfo-formulas in an equivalence class $[\theta]_\T$, we can extend the definition of the potential to the equivalence classes by setting $[p]([\theta]_\T) := p(\theta)$. We next want to show that applying any rule of the system $Y'$ on an equivalence class decreases the potential $[p]$. We show this at the level of pfo-formulas.
    \begin{claim}\label{clm:potentialdecrease}
        Applying any rule of $Y'$ on a pfo-formula $\theta$ to get a new pfo-formula $\theta'$ decreases the potential $p$. 
    \end{claim}
    \begin{proof}
        We consider the two rules $\rarr_{\splitdown}$ and $\rarr_{M}$ of $Y'$ individually. We start with the latter. So assume that we forget a quantifier in $\theta$, say $\forall x$, since the other case is completely analogous. Let $v$ be the node labeled with the quantifier we eliminate. Then in the definition of $p(\theta)$ we lose the term $\bar p(v) + \mu(v)$ which, by construction, is positive: either $\bar p(v)=0$ in which case $\mu(v)=1$ and clearly $\bar p(v) + \mu(v)>0$. Otherwise, $\bar p(v)> 0$ and by definition $\mu(v)=0$, so $\bar p(v)> 0$. We will show the other summands in the definition of $p(\theta)$ do not change when deleting $v$. To this end, observe that, since we can eliminate $v$, it does not bind any variables in any atoms. So $v$ is non-blocking for all other nodes $u$ labeled with quantifiers. But then for all such nodes $u$ we have that the non-blocked ancestor $u^*$ of $u$ does not change, unless $v= u^*$ in which case the new non-blocked ancestor of $u$ is the single child of $v$. Thus, in any case, the number of $\land$-nodes below $u^*$ and consequently $\bar p(u)$ remain unchanged. It follows that the only change from $p(\theta)$ to $p(\theta')$ is that we lose one positive summand and thus the potential decreases strictly.
        
        Now say that we use $\rarr_{\splitdown}$, say 
        \begin{align}\label{eq:1}
            \forall x (F_1 \wedge F_2) \rarr_{\splitdown} (\forall x F_1) \wedge (\forall x F_2)
        \end{align}
        (the existential case is again totally symmetric).
        Let $v$ be the node containing the quantification before the application and let $v_1$ and $v_2$ contain them afterwards. Let $a$ be the number of $\land$-nodes below $v$ before the application of the rule. Note that $\bar p(v)\ge a\ge 1$, so $\mu(v) = 0$. The $\land$-node above $v_1$ and $v_2$ is blocking for both of them, so $v_1$ and $v_2$ are their respective non-blocked ancestors and thus $\bar p(v_1) + \bar p(v_2) +1= a \le \bar p(v)$. We get that 
        \begin{align*}
            3\bar p(v) + \mu(v)& = 3 \bar p(v)\\
            & \ge 3a\\
            & = 3(\bar p(v_1) + \bar p(v_2) +1)\\
            & > 3\bar p(v_1) + \mu(v_1 )+ 3\bar p(v_2) + \mu(v_2).
        \end{align*}
        Here the last step holds because the function $\mu$ is bounded by $1$. It follows that, when substituting the summand for $v$ by those of $v_1$ and $v_2$, we decrease the overall sum. It remains to show that when going from $\theta$ to $\theta'$ none of the other summands in $p(\theta)$ increases. But this is easy to see: for all nodes labeled by a universal quantification, if the $\land$-node in $(\ref{eq:1})$ was blocking before, then it is still blocking afterwards, so the $\land$-nodes to consider in the computation of $\bar p(v)$ do not change. For a node $u$ labeled by an existential quantification, we have that the $\land$-node is not blocking before or after. Moreover, $v$ is blocking before the application of the rule if and only $v_1$ and $v_2$ are after. So, in any case, the non-blocked ancestor of any node $u\ne v$ stays the same and thus also $\bar p(u)$. This completes the proof of the claim.
    \end{proof}
    
    The proof of Lemma~\ref{lem:terminatingYp} is now easy to complete. The potential $[p]([\theta]_\T)$ of any equivalence class $[\theta]_\T$ is defined as $p(\theta)$ and thus by construction polynomial in $|\theta|$. Moreover, whenever we apply a rule of $Y'$ to $[\theta]_\T$, we apply it on a representative $\psi\in [\theta]_\T$ to get $\psi'$ and thus end up in the equivalence class $[\psi']_\T$. By Claim~\ref{clm:potentialdecrease}, we have $p(\psi) > p(\psi')$ and thus $[p]([\theta]_\T) = [p]([\psi]_\T) > [p]([\psi']_\T)$, so every application of a rule in $Y'$ decreases the potential $[p]$. Since the potential is by construction always a non-negative integer, any reduction chain starting in $[\theta]$ has length at most polynomial in $|\theta|$ which completes the proof.
\end{proof}

As the second ingredient for the proof of confluence for $Y'$ with the help of Lemma~\ref{lemma:newman}, we now show local confluence of $Y'$. 

\begin{lemmarep}\label{lem:locallyconfluentYp}
	The system $Y'$ is locally confluent.
\end{lemmarep}
\begin{proof}
    It will be convenient to prove a different, more technical statement which we first show implies local confluence of $Y'$. Remember that two equivalence classes $[\phi_1]_{\T\setminus \{N\}}$ and $[\phi_2]_{\T\setminus \{N\}}$ are called joinable w.r.t.~$\rarr_{\{\splitdown, M\}}$, in symbols $[\phi_1]_{\T\setminus \{N\}} \downarrow_{\{\splitdown, M\}}[\phi_2]_{\T\setminus \{N\}}$, if there exists a $\phi$ such that 
    \begin{align*}[\phi_1]_{\T\setminus \{N\}}\rarrstar_{\{\splitdown, M\}} [\phi]_{\T\setminus \{N\}}\larrstar_{\{\splitdown, M\}} [\phi_{2}]_{\T\setminus \{N\}}.\end{align*}
    
    \begin{claim}\label{clm:without-N}
        If for all standardized pfo-formulas $\phi_1$ and $\phi_2$ 
        where $\phi_1 \lrarrstar_{\T\setminus \{N\}}\phi_2$ we have that for all $\phi_1', \phi_2'$ with $\phi_1\rarr_{\{\splitdown, M\}} \phi_1'$ and $\phi_2\rarr_{\{\splitdown, M\}} \phi_2'$ that $[\phi_1']_{\T\setminus \{N\}} \downarrow_{\{\splitdown, M\}}[\phi_2']_{\T\setminus \{N\}}$, i.e.~$[\phi_1']_{\T\setminus \{N\}}$ and $[\phi_2']_{\T\setminus \{N\}}$ are joinable w.r.t.~$\rarr_{\{\splitdown, M\}}$, then $Y'$ is locally confluent.
    \end{claim}	
    \begin{proof}
        We assume the condition of the statement. Let $\psi, \psi_1, \psi_2$ be such that 
        \begin{align*}
            [\psi_1]_\T \larr_{\{\splitdown, M\}} [\psi]_\T \rarr_{\{\splitdown, M\}} [\psi_2]_\T.
        \end{align*} 
        We have to show that $[\psi_1]_\T$ and $[\psi_2]_\T$ are joinable w.r.t. $\rarr_{\{\splitdown, M\}}$. 
        
        We may assume w.l.o.g.~that there are $\theta_1, \theta_2\in [\psi]_\T$ such that $\theta_1 \rarr_{\{\splitdown, M\}} \psi_1$ and $\theta_2 \rarr_{\{\splitdown, M\}} \psi_2$ (otherwise, just choose another representative for $[\psi_1]_\T$ and $[\psi_2]_\T$).
        We have $\theta_1, \theta_2\in [\psi]_\T$, so $\theta_1 \lrarrstar_{\T} \theta_2$. By Proposition~\ref{prop:getridofN}, there are formulas $\phi_1, \phi_2$ such that 
        \begin{align*}
            \theta_1 \rarrstar_N \phi_1 \rarrstar_{\T\setminus \{N\}} \phi_2\rarrstar_N \theta_2
        \end{align*}
        Since all applications of $\rarr_N$ are invertible, we have in fact that $\theta_1\lrarrstar_N \phi_1$ and $\theta_2\lrarrstar_N \phi_2$. Let $\mathcal N^1 := N_1^1, \ldots, N^1_s$ be the sequence of renaming operations that gets us from $\theta_1$ to $\phi_1$ and let $\mathcal N^2:= N_1^2, \ldots, N_{s'}^2$ be that getting us from $\theta_2$ to $\phi_2$. 
        
        Remember that $\theta_1\rarr_{\{\splitdown, M\}}\psi_1$ and $\theta_2\rarr_{\{\splitdown, M\}}\psi_2$; call the rule applications $R_1$ and $R_2$ to get formulas $\phi_1'$ and $\phi_2'$, respectively. Since $\rarr_{\splitdown}$ and $\rarr_M$ are independent of the variables in the quantifications, we can apply the $R_1$ and $R_2$ to $\phi_1$ and $\phi_2$, respectively. Moreover, the application of $R_1$ commutes with that of $\mathcal N^1$, that is $\theta_1 \rarr_{R_1} \psi_1 \rarr_{N_1^1} \rho_1\rarr_{N_2^1} \ldots \rarr_{N_s^1} \rho_s=\phi_1'$ (where we tacitly delete renaming operations if a quantifier gets deleted by $R_1$) gets us to the same formula $\phi_1'$ as first renaming in $\theta_1$ according to $\mathcal N^1$ to get to $\phi_1$ and then applying $R_1$. Similarly, we get from $\theta_2$ to $\phi_2'$ by either applying first $R_2$ and then $\mathcal N^2$ or applying $R_2$ on $\phi_2$.
        
        Since $\phi_1\rarr_{\{\splitdown, M\}} \phi_1'$, $\phi_2\rarr_{\{\splitdown, M\}} \phi_2'$ and $\phi_1 \lrarrstar_{\T\setminus \{N\}} \phi_2$, by the assumption there is a formula $\psi$ with  \begin{align*}[\phi_1']_{T\setminus \{N\}} \rarrstar_{\{\splitdown, M\}} [\psi]_{T\setminus \{N\}} \larrstar_{\{\splitdown, M\}} [\phi_2']_{T\setminus \{N\}}.\end{align*} It follows that 
        \begin{align*}[\psi_1]_\T = [\phi_1]_\T\rarrstar_{\{\splitdown, M\}} [\psi]_\T\larrstar_{\{\splitdown, M\}} [\phi_{2}]_\T=  [\psi_2]_\T,\end{align*}
        so $[\psi_1]_\T$ and $[\psi_2]_\T$ are joinable and the claim follows.
    \end{proof}

    In the following, we will prove the assumption of Claim~\ref{clm:without-N}.
    Since we only have to consider two rules, it suffices to consider the different combinations of rule applications in 
    \begin{align}
        \label{eq:conf1}\phi_1 &\rarr_{\{\splitdown, M\}} \phi_1'\\
        \label{eq:conf2}\phi_2 &\rarr_{\{\splitdown, M\}} \phi_2'.
    \end{align}
    We consider three different cases.
    
    \textbf{Case 1 -- in (\ref{eq:conf1}) and (\ref{eq:conf2}) we use $\mathbf{\rarr_{M}}$:} With a short induction it is easy to see that if in $\phi_1$ a quantifier does not bind a variable, then this is also true for all other formulas in $[\phi_1]_{\T\setminus \{N\}}$. An analogous statement is true for quantifiers in $\phi_2$. We may assume that in~(\ref{eq:conf1}) and (\ref{eq:conf2}) different quantifiers are deleted since otherwise we directly end up in the same equivalence class with respect to $\lrarrstar_{\T\setminus \{N\}}$. Observe that deleting the same quantifier in $\phi_1$ or $\phi_2$ has no impact on the respective other quantifier. Then, in both cases, after deleting the first quantifier, the second one still does not bind any variable. It follows that after deleting both quantifiers in both $\phi_1$ and $\phi_2$ in any order, we end up in the same equivalence class with respect to $\lrarrstar_{\T \setminus \{N\}}$.
    
    \textbf{Case 2 -- in (\ref{eq:conf1}) we use $\mathbf{\rarr_{M}}$ and in (\ref{eq:conf2}) we use $\mathbf{\rarr_{\splitdown}}$:} Assume first that in (\ref{eq:conf1}) and (\ref{eq:conf2}) we apply the rules on the same quantifier. Then the newly introduced quantifiers in $\phi_2'$ do not bind any variables, so we can simply use $\rarr_{M}$ twice on $\phi_2'$ to delete them and get a formula in  $[\phi_1']_{\T\setminus \{N\}}$.
    
    If (\ref{eq:conf1}) and (\ref{eq:conf2}) are applied on different quantifiers, then it is evident that applying one of the operations first does not change the applicability of the second, and after applying both on both $\phi_1$ and $\phi_2$, we end up in formulas equivalent with respect to $\lrarrstar_{\T\setminus \{N\}}$.
    
    \textbf{Case 3 -- in (\ref{eq:conf1}) and (\ref{eq:conf2}) we use $\mathbf{\rarr_{\splitdown}}$:} If one application of $\rarr_{\splitdown}$ is on an $\exists$-quantifier while the other is on a $\forall$-quantifier, then we can reason as in Case 2. So assume in the following that both applications are on $\forall$-quantifiers; the case of $\exists$-quantifiers is completely symmetric. Let $v_1$ and $v_2$ be the respective root nodes of the corresponding subformulas in $\phi_1$, resp.~$\phi_2$, that we apply $\rarr{\splitdown}$ on (note that $v_1=v_2$ is possible here). Assume that both $v_1$ and $v_2$ bind a variable; otherwise we can use an application of $\rarr_{M}$ to delete the non-binding quantifier and proceed similarly to Case 2.
    
    By Proposition~\ref{prop:each-pfo-is-organized} we have that $\phi_1$ and $\phi_2$ are organized. Thus, for $i=1,2$ there is a largest holey sub-formula $F_i$ of $\phi_i$ that is a $\fove$-formula and that contains $v_i$. If $v_1$ does not lie in $F_2$, then clearly the two rule applications can be exchanged and treated as in Case 2. So assume in the remainder that $v_1$ and $v_2$ lie both in $F_1$ and $F_2$ and $[F_1]_{\T\setminus \{N\}} = [F_2]_{\T\setminus \{N\}}$. Since applications of $\rarr_{\splitdown}$ for a $\forall$-quantifiers cannot happen inside a $\fove$-formula, $v_1$ and $v_2$ must lie above a hole in $h_1$ in $F_1$ and $h_2$ in $F_2$, respectively. By Proposition~\ref{prop:T-to-organized}, the holes are the same as in $F_1$ and $F_2$ with respect to the formulas that are plugged into them up to $\lrarrstar_\T$, so somewhat abusing notation, we say that they are the same holes. So we have that both $h_1$ and $h_2$ appear both in $F_1$ and $F_2$. Consider first that case that $h_1$ and $h_2$ are different. If the quantification in $v_1$ binds a variable, then this variable can only appear in $h_1$. So w.l.o.g., using $\lrarrstar_{\T}$, we may assume that $v_1$ is directly above $h_1$ in both $F_1$ and $F_2$. Similarly, we may assume that $v_2$ is directly above $h_2$ in $F_1$ and $F_2$.
    Then we can again exchange the rule applications freely and proceed as in Case 2. Now assume  that $v_1$ and $v_2$ are above the same hole $h$ in $F_1$ and $F_2$. Note that out of the holes of $F_1, F_2$ the quantifiers $v_1$ and $v_2$ bind a variable only associated to $h$ and no other hole. So $v_1$ must be on a path from the root of $F_2$ to $h$ in $F_2$ and thus we can use $\rarr_{\pushdown}$ and $\rarr_{O}$ in $F_2$ to move $v_1$ directly on top of $v_2$. Thus we assume in the remainder that $F_2$ is in a form in which $v_1$ is the parent of $v_2$ which is the parent of $h$. We proceed accordingly if $v_1=v_2$.
    
    Since $\phi_1$ and $\phi_2$ are organized and $v_1$ and $v_2$ are pushed down, the hole $h$ must in $\phi_1$ and $\phi_2$ contain a $\exwe$-formula $\psi$ that has at least one $\land$-node and is thus not atomic. Let $F$ be the largest holey sub-formula of $\psi$ that is rooted in $h$. 
    
    In the remainder, we will need a better understanding of $[F]_{\T\setminus \{N\}}$. To this end, we define a relation $\sim$ on the holes of $F$ by saying that $h_1\sim h_2$ if and only if there is a variable $x$ associated to both $h_1$ and $h_2$ and there is an $\exists$-quantifier in $F$ binding $x$. Let $\approx$ be the transitive and reflexive closure of $\sim$. Note that $\approx$ is an equivalence relation on the holes of $F$. Moreover, this relation is the same for every formula $F'\in [F]_\T$.
    
    \begin{claim}\label{clm:equivalenceonholes}
        For every $F'\in [F]_\T$ and every equivalence class $C$ of $\approx$, there is an existential quantifier $u$ in $F'$ such that all holes in $C$ are below $u$ in $F'$.
    \end{claim}
    \begin{proof}
        We assume w.l.o.g.~that $C$ contains at least two holes. Let $\bar F'$ be the smallest subformula of $F'$ containing all holes in $C$. Then the root of $\bar F'$ is a $\land$-node and none of the two subformulas $\bar F'_1, \bar F'_2$ contains all holes in $C$. Let $C_1$ be the holes in $C$ appearing in $\bar F'_1$ and $C_2$ those in $\bar F'_2$. Then there must be a variable $x$ that is associated to a hole in $C_1$ and one in $C_2$ and they are bound by the same quantifier, because otherwise $C$ would not be an equivalence class. It follows that the quantifier binding $x$ in the holes must be an ancestor of the root of $\bar F'$ and thus all holes in $C$ lie below it.
    \end{proof}
    We rewrite $F$ with $\T$ as follows: first pull up all quantifiers to turn $F$ into prenex form (this is possible since $\phi_1$, $\phi_2$ and thus $F$ are standardized), then write the quantifier free part as 
    \begin{align}\label{eq:notyetcanonical}\bigwedge_C (\bigwedge_{h\in C} h),
    \end{align}
    where the outer conjunction goes over all equivalence classes $C$ of $\approx$.
    Then push all $\exists$-quantifiers of $F$ down as far as possible with $\rarr_{\pushdown}$. Note that since every quantifier binds a variable that only appear in one class $C$, all quantifiers can be pushed into the outer conjunction in (\ref{eq:notyetcanonical}). This gives a representation of $F$ as 
    \begin{align}\label{eq:canonical}
        F \equiv \bigwedge_{C} \psi_C,
    \end{align}
    where the conjunction is again over all equivalence classes with respect to $\approx$. By Claim~\ref{clm:equivalenceonholes} this representation is canonical up to $\lrarrstar_{\T\setminus \{N\}}$ inside the $\psi_C$ and all $\psi_C$ have an $\exists$-quantifier on top. Since the $\forall$-quantifiers $v_1$ and $v_2$ cannot be split over $\exists$-quantifiers, they cannot be moved into any formula in $[\psi_C]_{\T\setminus \{N\}}$. It follows that when using $\rarr_{\splitdown}$ on $v_1$ and $v_2$ in $\phi_1$ and $\phi_2$, we essentially split $v_1$ and $v_2$ over a conjunction of the $\psi_C$: we can split $v_1$ and $v_2$ exactly over all $\land$-nodes in (\ref{eq:canonical}). Moreover, we can split them over all those $\land$-nodes iteratively. From this we directly get that $[\phi_1']_{\T\setminus \{N\}}\downarrow_{\{\splitdown, M\}}[\phi_2']_{\T\setminus \{N\}}$ also in this case which completes the proof of Lemma~\ref{lem:locallyconfluentYp}.
\end{proof}

We now apply Lemma~\ref{lemma:newman} using Lemma~\ref{lem:terminatingYp} and Lemma~\ref{lem:locallyconfluentYp}, to directly get Theorem~\ref{thm:convergent}.

\section{Between rewriting and tree decompositions}
\label{sect:tds}

The aim of this section is to link tree decompositions and formula width.  We will then use our development to prove Theorem~\ref{thm:minimize-in-T}. 
In what follows, we focus on $\{ \exists, \wedge \}$-formulas, but our results
apply to the corresponding dual formulas, namely, $\{ \forall, \vee \}$-formulas.

Let $\phi$ be for now a standardized holey $\{ \exists, \wedge\}$-formula,
and let $a$ be an association for $\phi$.
We define the hypergraph of $(\phi,a)$ as the hypergraph
whose vertex set is $\bigcup_{i} a(i)$, where the union is over all holes $i$ appearing in $\phi$,
and whose edge set is $\{ a(i) ~|~ \mbox{$i$ appears in $\phi$} \} \cup \{ \free(\phi) \}$;
that is, this hypergraph has an edge corresponding to each hole of $\phi$, and an edge
made of the free variables of $\phi$.

\begin{example}
    Consider the holey formula $\phi := \exists z\ 1 \land 2$ and the association $a$ with $a(1):= \{x,z\}$ and $a(2):= \{y,z\}$. Then $\free(\phi) = \{x, y\}$, and thus the hypergraph of $(\phi,a)$ is the triangle on vertices $x,y,z$ so $H=(V,E)$ with $V=\{x,y,z\}$ and $E= \{\{x,y\}, \{x,z\}, \{y,z\}\}$.
\end{example}

The following theorem identifies a correspondence between the rules
in $\T \setminus \{N\}$ and the computation of minimum tree decompositions.
A related result appears as~\cite[Theorem 7]{DalmauKV02}.\footnote{
	Let us remark that, in the context of the present theorem,
	the presence of the rewriting rule $O$ is crucial:
Example~\ref{ex:role-of-O} shows that absence of this rule affects
the minimum width achieveable.
}

\begin{theorem}\label{thm:width-for-regions}
	Let $\phi$ be a standardized holey $\{ \exists, \wedge\}$-formula,
	and let $a$ be an association for~$\phi$.
	Let $H$ be the hypergraph of $(\phi, a)$.
	It holds that $\width([\phi]_{\T \setminus \{ N \} }) = \tw(H) + 1$.
	Moreover, there exists a polynomial-time algorithm that,
	given the pair $(\phi,a)$ and a tree decomposition $C$ of $H$, outputs a formula
	$\phi' \in [\phi]_{\T \setminus \{ N \}}$
	where $\width(\phi') \le \bagsize(C)$;
	when $C$ is a minimum width tree decomposition of $H$, 
	it holds that  $\width(\phi') = \bagsize(C)$.
\end{theorem}
\begin{proof}
    We first prove the equality $\width([\phi]_{\T \setminus \{ N \} }) = \tw(H) + 1$
    by proving two inequalities, and then discuss the algorithm.
    
    $(\geq)$: For this inequality, it suffices to prove that, given a formula
    $\phi' \in [\phi]_{\T \setminus \{ N \} }$, there exists a tree decomposition $C$ of $H$
    having $\bagsize(C) = \width(\phi')$.
    We view $\phi'$ as a rooted tree, and define $T$ to be a rooted tree with the same structure as
    the tree $\phi'$, so that $T$ contains a vertex for each subformula occurrence $\beta$
    of $\phi'$,
    and for $t_1, t_2 \in V(T)$, it holds that $t_2$ is a descendent of $t_1$ if and only if
    the subformula occurrence $\beta_2$ corresponding to $t_2$ is a subformula (in $\phi'$)
    of the subformula occurrence $\beta_1$ corresponding to $t_1$.
    For each vertex $t$ of $T$, we define $B_t = \free(\beta)$, where $\beta$ is the subformula
    occurrence corresponding to $t$.
    Define
    $C = (T, \{ B_t \}_{t \in V(T)})$.
    It is clear that the maximum size $|B_t|$ of a bag $B_t$ is equal to
    $\width(\phi')$, and it is straightforward to verify that $C$
    is a tree decomposition of $H$.

    $(\leq)$: For this inequality, 
    we show how to algorithmically pass from a tree decomposition $C = (T, \{ B_t \})$ of $H$ to a formula
    $\phi' \in [\phi]_{\T \setminus \{ N \}}$
    where $\width(\phi') \leq \bagsize(C)$.
    We view $T$ as a rooted tree by
    picking a root $r$ having $\free(\phi) \subseteq B_r$.
    We define the passage inductively.  For each hole $i$ of $\phi$, let~$t_{i}$
    be a vertex of $T$ such that $a(i) \subseteq B_{t_i}$ which exists by the edge coverage property of tree decompositions.
    
    As a basic case, consider the situation where each such vertex $t_i$ is equal to the
    root $r$ of $T$.  Then, $B_r$ contains all variables occurring in $a(i)$, over each hole $i$
    of $\phi$, and thus $\phi$ itself has $\width(\phi) \leq |B_r| \leq \bagsize(C)$.
    This basic case covers each situation where $T$ contains a single vertex.
    
    When this basic case does not occur, let $r_1, \ldots, r_k$ be each child of the root
    such that the subtree of $T$ rooted at $r_j$, denoted by $T_j$, contains a vertex
    of the form $t_i$.  (Since we are not in the basic case, we have $k \geq 1$.)
    Let $H_j$ be the set containing each hole $i$ such that $t_i$ occurs in $T_j$.
    Let $V_j$ be the set containing all variables of the holes in $H_j$,
    that is, define $V_j = \bigcup_{i \in H_j} a(i)$.
    Let $U_j$ be the subset of $V_j$ containing each variable $v \in V_j$ that is not
    in $B_r$.
    Due to the connectivity of $C$, the sets $U_j$ are pairwise disjoint;
    since $\free(\phi) \subseteq B_r$ (by our choice of the root $r$), 
each variable in a set $U_j$ is quantified in $\phi$.
    Let $U$ be the set containing each variable that is quantified in $\phi$
    that is not in $\bigcup_{j=1}^k U_j$, and let $H$ be the set containing each hole in $\phi$
    that is not in $\bigcup_{j=1}^k H_j$.
    Since $\phi$ is standardized, by repeatedly applying the $\pushup$ rule, we can
    rewrite $\phi$ into prenex form; by using the $O$ rule on the existential quantifications
    and the $A$ and $C$ rules on the conjunctions,
    we can rewrite $\phi$ in the form
    $$\exists U \exists U_1 \ldots \exists U_k ((\bigwedge H) \wedge (\bigwedge H_1) \wedge \cdots \wedge (\bigwedge H_k)),$$
    where for a set of variables $S$, we use $\exists S$ to denote the existential
    quantification of the elements in $S$ in some order, and for a set of holes $I$,
    we use $\wedge I$ to denote the conjunction of the elements in $I$ in some order.
    By applying the $A$ and $C$ rules on the conjunctions and by using the $\pushdown$ rule,
    we can then rewrite the last formula as
    $$\exists U ( (\bigwedge H) \wedge (\exists U_1 \bigwedge H_1) \wedge \cdots \wedge (\exists U_k \bigwedge H_k)).$$
    For each $j = 1, \ldots, k$, let $\phi_j$ denote $\exists U_j \bigwedge H_j$.
    We claim that $T_j$ and the bags $(B_t \cap V_j)$ (over $t \in V(T_j)$),
    constitute a tree decomposition $C_j$ of the hypergraph of $\phi_j,a$:
    the coverage conditions hold by the definitions of $H_j$ and $T_j$;
    the connectivity condition is inherited from $C$; and, the bag
    $B_{r_j} \cap V_j$ contains each variable in
    $\free(\phi_j) = V_j \setminus U_j$, since each such variable appears
    both in $B_r$ (by definition of $U_j$) and in a bag of $C_j$.
    By invoking induction, each formula $\phi_j$ can be rewritten, via the $\T \setminus \{N\}$ rules,
    as a formula $\phi'_j$ where $\width(\phi'_j) \leq \bagsize(C)$.
    We define the desired formula as
    $$\phi' = \exists U ( (\bigwedge H) \wedge \phi'_1 \wedge \cdots \wedge \phi'_k).$$
    To verify that $\width(\phi') \le \bagsize(C)$, we need only verify that
    the formula
    $F = ( (\bigwedge H) \wedge \phi'_1 \wedge \cdots \wedge \phi'_k)$
    has at most $\bagsize(C)$ free variables.
    This holds since each subformula of $\phi$ is either a subformula of a formula $\phi'_j$,
    or has its free variables contained in $B_r$.
    We conclude that
    $\width(\phi') \le \bagsize(C)$.
    
    (Algorithm): We explain how
    the rewriting of $\phi$ into $\phi'$, that was just presented,
    can be performed in polynomial time.
    It suffices to verify that the number of invocations of induction/recursion
    is polynomially bounded:
    aside from these invocations, the described rewriting can clearly be carried out
    in polynomial time.  But at each invocation, the recursion is applied to a
    tree decomposition whose root is a child of the original tree decomposition's root,
    along with a formula that is smaller than the original formula.
    Hence, the total number of invocations is bounded above by the number of vertices
    in the original tree decomposition and thus the size of the formula and its hypergraph as discussed in Section~\ref{sct:prelims}.
    We have that $\width(\phi') \le \bagsize(C)$.
    When $C$ is a minimum width tree decomposition of $H$, the initial inequality
    proved rules out $\width(\phi') < \bagsize(C)$, so in this case,
    we obtain $\width(\phi') = \bagsize(C)$.
\end{proof}

The following is a consequence of Theorem~\ref{thm:width-for-regions},
and Proposition~\ref{prop:getridofN}.

\begin{corollary}
\label{cor:exists-wedge-tw}
Let $\phi$ be a standardized $\{ \exists, \wedge \}$-sentence, and
let $H$ be the hypergraph of $\phi$ (defined as having a hyperedge
$\{ v_1, \ldots, v_k \}$ for each atom $R(v_1, \ldots, v_k)$ of $\phi$).
It holds that 
\[\width([\phi]_{\A }) = \width([\phi]_{\T }) = \width([\phi]_{\T \setminus \{ N \} }) = \tw(H) + 1.\]
\end{corollary}
\begin{proof}
    We get $\width([\phi]_{\T \setminus \{ N \} }) = \tw(H) + 1$ from Theorem~\ref{thm:width-for-regions}. Using Proposition~\ref{prop:getridofN} and the fact that renaming variables clearly does not change the width, we get $\width([\phi]_{\T }) = \width([\phi]_{\T \setminus \{ N \} })$. Finally, we have $\A = \T \cup \{ \splitdown, \splitup, M \}$. But $\splitdown$ and $\splitup$ can never be used in an $\{ \exists, \wedge \}$-sentence and $M$ does not change the width of any formula, we have $\width([\phi]_{\A }) = \width([\phi]_{\T })$.
\end{proof}    

We now have everything in place to prove Theorem~\ref{thm:minimize-in-T} which we restate for the reader's convenience.
\theoremminimizeinT*
\begin{proof}%
    Let us first sketch how to employ Theorem~\ref{thm:width-for-regions} to prove Theorem~\ref{thm:minimize-in-T}.
    The proof relies on the fact that, via Proposition~\ref{prop:T-to-organized}, we can consider the regions of the formula $\theta$ independently. Moreover, Theorem~\ref{thm:width-for-regions} allows us to minimize the width of the individual regions. We now give the technical details.
    
    Let $\theta$ be a given pfo-formula. 
If $\theta$ is not standardized,
we compute a standardized version of $\theta$, that 
is, a pfo-formula $\theta'$ where $\theta \rarrstar_N \theta'$,
and replace $\theta$ with this standardized version $\theta'$;
by appeal to Proposition~\ref{prop:getridofN}, 
it suffices to find a pfo-formula having minimum width 
among all formulas in $[\theta]_{\T\setminus \{ N \} }$.
We compute $\theta^+$ as follows: if $\theta$ is a $\fove$-formula without any holes, we compute its hypergraph and compute a tree decomposition of minimal width for it. Then we use Theorem~\ref{thm:width-for-regions} on the tree decomposition to get the formula $\theta^+$. The case of $\exwe$-formulas is treated analogously.
    
    If $\theta$ is neither a $\fove$-formula nor a $\exwe$-formula, we proceed as follows: using Proposition~\ref{prop:T-to-organized}, we first compute the representation $\phi\llb \phi_1, \ldots, \phi_k\rrb$ for $\theta$ where $\phi$ is the region of~$\theta$ containing the root. Assume that $\phi$ is a $\exwe$-formula; the other case is completely analogous. Let $a$ be the association that for every $i\in [k]$ maps $i\mapsto \free(\phi_i)$. For every $i\in [k]$, we recursively compute minimum width formulas $\phi_i^+\in [\phi_i]_{\T \setminus \{N\}}$. We then compute $\phi^+$ using a minimum width tree decomposition of the hypergraph of $(\phi, a)$ and applying Theorem~\ref{thm:width-for-regions}. Finally, we compute $\theta^+ := \phi^+\llb \phi_1^+, \ldots, \phi_k^+\rrb$ which is well-defined since for every $i\in [k]$ the hole $i$ of $\phi$ and thus also $\phi^+$ is associated the variables $\free(\phi_i) = \free(\phi_i^+)$. This completes the description of the algorithm.
    
    Since in every recursive step the formula size decreases and the number of recursive calls is bounded by the number of regions, the algorithm clearly runs in polynomial time up to the calls of the oracle for treewidth computation. It remains to show that $\theta^+$ has minimum width. By way of contradiction, assume this were not the case. Then there is a recursive call on a formula $\bar \theta$ such that $\bar \theta$ is of minimal size such that there is a formula $\theta'\in [\bar \theta]_{\T\setminus \{N\}}$ with $\width(\theta') < \width(\bar \theta^+)$. Assume first that $\bar \theta$ has only atomic formulas in its holes. Then, by construction, we have $\width(\bar \theta^+) = \bagsize(T, \{B_t\}_{t\in V(T)}) = \tw(H)+1$ where $(T, \{B_t\}_{t\in V(T)})$ is a minimum width tree decomposition of the hypergraph~$H$ of $\bar \theta$. Since, by assumption, $\width(\theta') < \width(\bar \theta^+)$, we get that 
    \begin{align*}
        \width([\bar \theta]_{\T\setminus \{N\}}) \le \width(\theta') < \width(\bar \theta^+) = \tw(H)+1,
    \end{align*}
    which contradicts Theorem~\ref{thm:width-for-regions}.
    
    Now assume that $\bar \theta$ is of the form $\phi\llb \phi_1, \ldots, \phi_k\rrb$ where some of the $\phi_i$ are not atomic. Then, we have $\bar \theta^+ := \phi^+\llb \phi_1^+, \ldots, \phi_k^+\rrb$. Assume that there is a smaller width formula $\theta'$. Then $\theta'$ can by Proposition~\ref{prop:T-to-organized} be written as $\theta' := \phi' \llb \phi_1', \ldots, \phi_k'\rrb$ where for every $i\in [k]$ we have $\phi_i' \lrarrstar_{\T \setminus \{ N \}} \phi_i^+$. We assumed that the recursive calls for the $\phi_i$ computed minimum width formulas $\phi_i^+$, so $\width(\phi_i^+) \le \width(\phi_i')$. We have $\width(\bar \theta^+) = \max(\width(\phi^+), \max\{\width(\phi_i^+)\mid i\in [k]\})$. Moreover, $\width(\theta') \ge \max(\{\width(\phi_i')\mid i\in [k]\}) \ge \max(\{\width(\phi_i^+)\mid i\in [k]\})$. It follows that 
    \begin{align*}
        \width(\theta') < \width(\bar \theta^+) \le \max(\width(\phi^+), \width(\theta')),
    \end{align*}
    which can only be true if $\width(\phi^+) > \width(\theta') \ge \width(\phi_i^+)$ for all $i\in [k]$. We get that $\width(\bar \theta^+) = \width(\phi^+)$. Using that $\width(\theta') \ge \width(\phi')$, we get $\width(\phi^+) > \width(\phi')$. Arguing as before, we have $\width(\phi^+) = \tw(H)+1$ where $H$ is the hypergraph of $\phi$ and thus 
    \begin{align*}
        \width([\phi]_{\T\setminus \{N\}}) \le \width(\phi') < \width(\phi^+) = \tw(H)+1,
    \end{align*}
    which is again a contradiction with Theorem~\ref{thm:width-for-regions}.
    
    Thus, in every case, $\theta^+$ is a minimum width formula in $[\theta]_{\T \setminus \{ N \}}$, which completes the proof of Theorem~\ref{thm:minimize-in-T}.
\end{proof}

\section{Monotonicity}
\label{sect:monotonicity}

In this section, we prove monotonicity of the gauged system $G'$.
We begin by proving a theorem concerning the width, up to rules in $\T \setminus \{N\}$,
of holey $\{ \exists, \wedge \}$-formulas that are conjunctive.  This will allow us to understand
how to use these rules to minimize a formula that permits application of the splitdown rule. 

\begin{theorem}
	\label{thm:conjunction-on-top}
	Let $\theta$ be a standardized holey $\{ \exists, \wedge \}$-formula
	having the form $\phi \wedge \phi'$, and let~$a$ be an association for $\theta$.
	It holds that
	$\width([\phi \wedge \phi']_{\T \setminus \{ N \}})$
is equal to
\[
	\max(
	\width([\phi ]_{\T \setminus \{ N \}}),
	\width([\phi']_{\T \setminus \{ N \}}),
	|\free(\phi) \cup \free(\phi')|
	).\]
\end{theorem}
\begin{proof}
    Let $H$ be the hypergraph of $(\theta, a)$.
    
    ($\geq$):
    By Theorem~\ref{thm:width-for-regions} and the definition of treewidth,
    it suffices to show that, given a minimum tree decomposition $C$ of $(\theta,a)$,
    the value $\bagsize(C)$ is greater than or equal to each of the three values whose maximum is taken
    in the statement.
    Since $\free(\theta) = \free(\phi) \cup \free(\phi')$ is an edge of $H$,
    there exists a bag of $C$ containing $\free(\theta)$,
    and thus $\bagsize(C) \geq |\free(\theta)|$.
    It remains to show that $\bagsize(C) \geq \width([\phi]_{\T \setminus \{ N \}})$.
    (An identical argument shows that the corresponding statement for $\phi'$ holds.)
    By Theorem~\ref{thm:width-for-regions}, it suffices to give
    a tree decomposition $C^-$ of the hypergraph $H^-$ of $(\phi,a)$
    such that $\bagsize(C) \geq \bagsize(C^-)$.
    Denote $C$ by $(T,\{B_t\}_{t \in V(T)})$.
    Define $C^-$ as $(T,\{B^-_t\}_{t \in V(T)})$,
    where $B^-_t = B_t \cap V(H^-)$ for each $t \in V(T)$.  That is, we define new bags
    by restricting to the elements in $V(H^-)$.
    It is straightforward to verify that $C^-$ gives the desired tree decomposition.
    
    ($\leq$):
    Let $H$, $H'$, and $H^+$ be the hypergraphs of $(\phi,a)$, of $(\phi',a)$, and of $(\phi \wedge \phi',a)$,
    respectively.
    By Theorem~\ref{thm:width-for-regions}, it suffices to show that,
    given a minimum tree decomposition $C = (T,\{B_t\})$ of $H$
    and a minimum tree decomposition $C' = (T',\{B'_{t'}\})$ of $H'$,
    there exists a tree decomposition $C^+$ of $H^+$ such that
    $\bagsize(C^+) \leq \max(\bagsize(C),\bagsize(C'),|\free(\phi) \cup \free(\phi')|)$.
    By renaming vertices if necessary, we may assume that $V(T)$ and $V(T')$ are disjoint.
    Let $r \in V(T)$ and $r' \in V(T')$ be vertices where $B_r \supseteq \free(\phi)$
    and $B_{r'} \supseteq \free(\phi')$, and let $r^+$ be a fresh vertex that is not in $V(T) \cup V(T')$.
    We create a tree decomposition $C^+ = (T^+, \{ B^+_{t^+} \}_{t^+ \in V(T^+)})$
    where $V(T^+) = V(T) \cup V(T') \cup \{ r^+ \}$ and 
    $E(T^+) = E(T) \cup E(T') \cup \{ \{ r^+, r \}, \{ r^+, r' \} \}$.
    The bags are defined as follows: $B^+_{r^+} = \free(\phi) \cup \free(\phi')$;
    $B^+_t = B_t$ for each $t \in V(T)$, and $B^+_{t'} = B'_{t'}$ for each $t' \in V(T')$.
    So, the tree decomposition $C^+$ is created by taking $C$ and $C'$, and by adding a new vertex
    $r^+$ whose bag is as described and which is adjacent to just $r$ and $r'$.
    It is straightforward to verify that $C^+$ is a tree decomposition of $H^+$.
    We briefly verify connectivity as follows.
    Consider a variable $v$ in $V(H^+)$.  If $v$ is quantified in $\theta$, then, since $\theta$ is standardized,
    it occurs in just one of $\phi,\phi'$, and connectivity with respect to $v$ is inherited
    by connectivity within $C$ or $C'$, respectively.
    Otherwise, the variable $v$ is in $\free(\phi) \cup \free(\phi')$.
    If the variable $v$ occurs in $\free(\phi)$, it is in $B_r$, and if the variable $v$
    occurs in $\free(\phi')$, it is in $B_{r'}$; in all cases, connectivity is inherited,
    since $v$ occurs in $B_{r^+}$.
\end{proof}

We can use Theorem~\ref{thm:conjunction-on-top} to show Theorem~\ref{thm:monotone}.
Remember that 
$$G' = (\pfo, \to_{\{\splitdown,M\}}, \wid) \div \lrarrstar_{\T},$$ where $\wid$ is the function that 
maps a formula to its width. Then Theorem~\ref{thm:monotone} is the following statement.

\theoremmonotone*
\begin{proof}%
    Let $C_1, C_2$ be $\T$-equivalence classes such that
    $C_1 \rarr_{M} C_2$
    or $C_1 \rarr_{\splitdown} C_2$.
    We need to show that $\width(C_1) \geq \width(C_2)$.
    We argue this for each of the two rules, in sequence.
    
    ($M$): Since $C_1 \rarr_{M} C_2$, there exist formulas
    $\theta_1 \in C_1$, $\theta_2 \in C_2$
    such that $\theta_1 \rarr_{M} \theta_2$.
    We have that the rule $\rarr_{M}$ is monotone on individual formulas,
    that is, for formulas $\beta_1, \beta_2$, the condition $\beta_1 \rarr_{M} \beta_2$
    implies $\width(\beta_1) \geq \width(\beta_2)$.
    It thus suffices to show that
    there exists a formula $\theta_1^- \in C_1$
    having minimum width among all formulas in $C_1$
    and a formula $\theta_2^- \in C_2$ such that $\theta_1^- \rarr_{M} \theta_2^-$.
    We explain how to derive such a pair $\theta_1^-, \theta_2^-$ from $\theta_1, \theta_2$.
    
    Let $\theta_1^-$ be any minimum width element of $C_1$.
    We have that $\theta_1^-$ can be obtained from $\theta_1$ via a sequence of $\T$ rules.
    Let $Q x$ be the quantification in $\theta_1$ that can be removed via the $M$ rule.
    Apply the sequence of rules one by one to both $\theta_1$ and $\theta_2$, but without applying to $\theta_2$
    any rules that involve the removable quantification $Q x$ but moving other quantifiers over $Qx$ when necessary with the $O$ rule; note that this is possible because the only difference between $\theta_1$ and $\theta_2$ is that $Q x$ does not appear in $\theta_2$. When this transformation is done, we end up in formulas $\theta_1'$, resp., $\theta_2'$. It is straightforwardly
    verified that the relationship $\rarr_{M}$ is preserved, and thus the desired pair of formulas
    with $\theta_1^- \rarr_{M} \theta_2^-$ is obtained.
    
    ($\splitdown$): Paralleling the previous case, we assume that we have formulas
    $\theta_1 \in C_1$, $\theta_2 \in C_2$
    such that $\theta_1 \rarr_{\splitdown} \theta_2$, and we aim to show that
    there exists a formula $\theta_1^- \in C_1$
    having minimum width among all formulas in $C_1$
    and a formula $\theta_2^- \in C_2$ such that $\theta_1^- \rarr_{\splitdown} \theta_2^-$.
    
    Let $\theta_1^m$ be a minimum width element of $C_1$; we have that $\theta_1 \rarrstar_{\T} \theta_1^m$.
    By appeal to Proposition~\ref{prop:getridofN} and the fact that the rules in $\T \setminus \{ N \}$
    preserve the property of being standardized, we obtain formulas $\theta_1^n, \theta_1^+$
    such that $\theta_1 \rarrstar_{N} \theta_1^n \rarrstar_{\T \setminus \{ N \}} \theta_1^+$
    where $\theta_1^+$ is equivalent to $\theta_1^m$ up to renaming, and hence has minimum width,
    and where $\theta_1^n, \theta_1^+$ are standardized.
    By applying the renaming rules witnessing  $\theta_1 \rarrstar_{N} \theta_1^n$
    to $\theta_2$, we obtain a formula $\theta_2^n$ such that
    $\theta_1^n \rarr_{\splitdown} \theta_2^n$.
    We have that $\theta_1^n \rarrstar_{\T \setminus \{ N \}} \theta_1^+$,
    so to minimize the width of $\theta_1^n$ using the $\T$ rules, it suffices to do so using
    the $\T \setminus \{ N \}$ rules.
    Consider the decomposition of $\theta_1^n$ into regions,
    and let $\phi \llb \psi_1, \ldots, \psi_k \rrb$ be a subformula of $\theta_1^n$
    such that $\phi$ contains the quantifier to which $\splitdown$ can be applied.
    Up to duality and up to notation, we assume that $\phi$ is a holey $\fove$-formula
    containing $\forall y 1$ as a subformula (where $1$ is a hole), that $\psi_1$ is a $\exwe$-organized formula
    of the form $\beta \wedge \beta'$, and that $\forall y$ is the quantification
    to which $\splitdown$ can be applied.  We argue that
    the subformula $\phi \llb \psi_1, \ldots, \psi_k \rrb$ can be rewritten by $\T \setminus \{ N \}$
    rules into a width-minimal form while keeping the $\forall y$ quantification
    above the outermost conjunction $\wedge$ of $\psi_1$, which suffices,
    since the same rewriting could be applied to $\theta_2^n$ to yield the desired formulas
    $\theta_1^-, \theta_2^-$.  The desired rewriting of $\phi$ exists: the variable $y$ occurs
    free only in $\psi_1$, since $\theta_1^n$ is standardized, and hence can be maintained right on top of $1$.
    The desired rewriting of $\psi_1$ exists, because Theorem~\ref{thm:conjunction-on-top} shows that each of the conjuncts of the outermost conjunction can be minimized individually.
\end{proof}

\section{Conclusion}

We have seen that one can optimally minimize the width of positive first-order formulas with respect to the application of syntactic rules. The algorithm we have given runs in polynomial time, up to the use of an algorithm 
computing minimum tree decompositions of hypergraphs. 
While the computation of such decompositions is known to be
$\mathsf{NP}$-hard, there exist FPT-algorithms parameterized by the treewidth
(see for example~\cite{BodlaenderDDFLP13,Bodlaender96}) and it follows directly that our
width minimization algorithm runs in FPT-time parameterized by the optimal width.

From our techniques, it can be seen that equivalence up to the syntactical rules that we study can also be checked for pfo-formulas: given two pfo-formulas $\phi_1, \phi_2$, compute their normal forms with the algorithm of Theorem~\ref{thm:minimize-in-T}. Then, $\phi_1$ and $\phi_2$ are equivalent if there is an isomorphism of the structure of their regions---which is easy to check since the regions are organized in a tree shape---and the regions mapped onto each other by this isomorphism are equivalent. The latter can be verified by standard techniques for conjunctive queries~\cite{ChandraM77} which yields the algorithm for checking equivalence of pfo-formulas under our syntactical rules and, in the positive case, also allows extracting a sequence of rule applications that transforms $\phi_1$ to $\phi_2$.

We close by discussing an open issue: it would be interesting to extend the set of rules we allow in the width minimization. A particularly natural addition would be the distributive rules stating that $\land$ distributes over $\lor$, and vice-versa. To see why this operation would be useful, consider as an example the formula
\begin{align*}
    \phi:= \forall x_1 \ldots \forall x_n \left((S_1(x_1) \land \ldots \land S_n(x_n))\lor (T_1(x_1) \land \ldots \land T_n(x_n)) \right).
\end{align*}
With the rules that we have considered here, one can only permute the universal quantifications and use the tree decomposition rules on the conjunctions. Neither of these operations decreases the width of the formula which is $n$. With the distributivity rule, we can ``multiply out'' the two conjunctions to get
\begin{align*}
    \phi':= \forall x_1 \ldots \forall x_n \bigwedge_{i,j\in [n]} (S_i(x_i) \lor T_j(x_j)).
\end{align*}
Now we can split down the quantifiers over the conjunction ($\splitdown$), remove most of them (M) and push down ($\pushdown$) to end up with
\begin{align*}
    \phi'' = \bigwedge_{i,j\in [n], i\ne j} \left((\forall x_i S_i(x_i))\lor (\forall x_j T_j(x_j))\right) \land \bigwedge_{i\in [n]} \left(\forall x_i (S_i(x_i)\lor T_i(x_i))\right)
\end{align*}
which has width $1$. So distributivity has the potential of substantially decreasing the width of formulas.
 However, note that allowing this operation would necessarily change the form of the results we could hope for: it is known that there are formulas whose width minimization leads to an unavoidable exponential blow-up in the formula size~\cite{BerkholzC19}, and it is 
verifiable that this is also true when only allowing syntactical rewriting rules including distributivity. Note that this is very different from our setting where the rewriting rules that we consider may only increase the formula size by adding a polynomial number of quantifiers by the splitdown rule. Thus, when adding distributivity to our rule set, we cannot hope for polynomial time algorithms up to treewidth computation, as we showed in the present article. Still, it would be interesting to understand how distributivity changes the rewriting process beyond this and if there are algorithms that run polynomially, say, in the input and output size.

\paragraph{\bf Acknowledgements.} The authors thank Carsten Fuhs, Sta\'s Hauke, Moritz M\"uller, and Riccardo Treglia for useful feedback and pointers.
The first author acknowledges the support of EPSRC grants
EP/V039318/1   and
EP/X03190X/1.

\bibliographystyle{alphaurl}
\bibliography{rewriting}
\end{document}